\documentclass[a4paper,UKenglish,cleveref, autoref, thm-restate]{lipics-v2021}

\usepackage[utf8]{inputenc}
\usepackage{microtype,verbatim,dsfont}
\usepackage{graphicx}
\graphicspath{{figures/}}
\usepackage{caption}
\usepackage{subcaption}
\usepackage{comment}
\usepackage{mwe}
\usepackage{cite}
\usepackage{complexity}
\usepackage{todonotes}
\usepackage{amsmath, amsfonts}
\usepackage{enumerate}

\usepackage{float}
\usepackage{algorithm}
\usepackage{algpseudocode}
\algrenewcommand{\algorithmiccomment}[1]{// #1}

\newcommand{\crs}[1]{CRS($#1$)}
\newcommand{\crsv}[1]{CRSV($#1$)}

\newcommand{\cms}[1]{CMS($#1$)}

 \title{
Splitting Vertices in 2-Layer Graph Drawings
}

\author{Reyan~Ahmed}{Colgate University, Hamilton, NY, USA}{}{}{}

\author{Patrizio~Angelini}{John Cabot University, Rome, Italy}{}{}{}

\author{Michael~A.~Bekos}{University of Ioannina, Ioannina, Greece}{}{}{}

\author{Giuseppe~Di~Battista}{Roma Tre University, Rome, Italy}{}{}{}

\author{Michael~Kaufmann}{Universit\"at T\"ubingen, T\"ubingen, Germany}{}{}{}

\author{Philipp~Kindermann}{Universität Trier, Trier, Germany}{}{}{}

\author{Stephen~Kobourov}{University of Arizona, Tucson, AZ, USA}{}{}{}

\author{Martin~N\"ollenburg}{TU Wien, Vienna, Austria}{}{}{}

\author{Antonios~Symvonis}{NTUA, Athens, Greece}{}{}{}

\author{Ana\"is~Villedieu}{TU Wien, Vienna, Austria}{}{}{}

\author{Markus~Wallinger}{TU Wien, Vienna, Austria}{}{}{}

\authorrunning{R. Ahmed et al.}

\Copyright{R. Ahmed, P. Angelini, M. A. Bekos, G. D. Battista, M. Kaufmann, P. Kindermann, S. Kobourov, M. N\"ollenburg, A. Symvonis, A. Villedieu, M. Wallinger}



\ArticleNo{1}

\begin{document}
\nolinenumbers

\maketitle


\begin{abstract}
Bipartite graphs model the relationships between two disjoint sets of entities in several applications and are naturally drawn as 2-layer graph drawings. In such drawings, the two sets of entities (vertices) are placed on two parallel lines (layers), and their relationships (edges) are represented by segments connecting vertices.
Methods for constructing 2-layer drawings often try to minimize the number of edge crossings. 
We use vertex splitting to reduce the number of crossings, by replacing selected vertices on one layer by two (or more) copies and suitably distributing their incident edges among these copies.
We study several optimization problems related to vertex splitting, either minimizing the number of crossings or removing all  crossings with fewest splits. While we prove that some variants are \NP-complete, we obtain polynomial-time algorithms for others. 
We run our algorithms on a benchmark set of bipartite graphs representing the relationships between human anatomical structures and cell types.

\end{abstract} 

\newpage

\section{Introduction}
\label{sec:introduction}
Multilayer networks are used in many applications to model complex relationships between different sets of entities in interdependent subsystems~\cite{mragkp-vamn-21}. When analyzing and exploring the interaction between two such subsystems $S_t$ and $S_b$, bipartite or 2-layer networks arise naturally. The nodes of the two subsystem are modeled as a bipartite vertex set $V=V_t \cup V_b$ with $V_t \cap V_b = \emptyset$, where $V_t$ contains the vertices of the first subsystem $S_t$ and $V_b$ those of $S_b$. The inter-layer connections between $S_t$ and $S_b$ are modeled as an edge set $E \subseteq V_t \times V_b$, forming a bipartite graph $G=(V_t \cup V_b, E)$. Visualizing this bipartite graph $G$ in a clear and understandable way is then a key requirement for designing tools for visual network analysis~\cite{pfhlev-mvelbg-18}. 

A layered graph is a natural representation of the relationships between different genomic components, such as cell types and genes/proteins.
Graphs of this type can be obtained using massively parallel sequencing in bulk  at the single-cell level~\cite{borner2021anatomical_bibtex}. 
B{\"o}rner et al.~\cite{borner2021anatomical_bibtex} studied different tasks expected to be performed on such graphs (e.g., finding neighbors, finding shared neighbors, etc.).


In a \emph{2-layer graph drawing} of a bipartite graph 
the vertices are drawn as points on two distinct parallel lines $\ell_t$ and $\ell_b$, and edges are drawn as straight-line segments~\cite{EadesW94}. The vertices in $V_t$ (\emph{top vertices}) lie on $\ell_t$ (\emph{top layer}) and those in $V_b$ (\emph{bottom vertices}) lie on $\ell_b$ (\emph{bottom layer}).
In addition to direct applications of 2-layer networks for modeling the relationships between two communities as mentioned above~\cite{pfhlev-mvelbg-18}, such drawings also occur in tanglegram layouts for comparing phylogenetic trees~\cite{szh-trptn-11} or as components in layered drawings of directed graphs~\cite{SugiyamaTT81}. 

Purchase~\cite{purchase97_bibtex} has studied different graph layout properties that have an impact on human understanding. In particular, it is known that edge crossings can significantly impact graph-based tasks, which is why minimizing the number of edge crossings (or  removing all edge crossings altogether) is of interest. 

The primary optimization goal for 2-layer graph drawings is to find permutations of one or both vertex sets $V_t$, $V_b$ to minimize the number of edge crossings.
While the existence of a crossing-free 2-layer drawing can be tested in 
linear time~\cite{emw-oacp-86}, the
crossing minimization problem is \NP-complete 
even if the permutation of one layer is given~\cite{EadesW94}. 
Hence, both fixed-parameter 
algorithms~\cite{KobayashiT15nourl} and approximation algorithms~\cite{DemestrescuF01} have been published.
Further, graph layouts on two layers have also been widely studied in the area of graph drawing beyond planarity~\cite{dlm-sgdbp-19}.
However, from a practical point of view, minimizing the number of crossings in 2-layer drawings may still result in visually complex drawings~\cite{jm-2scmpeha-97}. 

Hence, in this paper, as an alternative approach to construct readable 2-layer drawings, 
we study vertex splitting~\cite{em-vtl-95,Liebers-survey01,ekkllm-pstg-18,Knauer2016}. 
The \emph{vertex-split} operation (or \emph{split}, for simplicity) for a vertex $v$ deletes $v$ from $G$, adds two new copies $v_1$ and $v_2$ (in the original vertex subset of $G$), and distributes the edges originally incident to $v$ among the two new vertices $v_1$ and $v_2$. 
Placing $v_1$ and $v_2$ independently in the 2-layer drawing can in turn reduce the number of crossings. 

Vertex splitting has been studied in the context of the \emph{splitting number} of an arbitrary graph $G$, which is the smallest number of vertex-splits needed to transform $G$ into a planar graph. The splitting number problem is \NP-complete, even for cubic graphs~\cite{faria_splitting_2001nourl}, but the splitting numbers of complete and complete bipartite graphs are known~\cite{HartsfieldJR85nourl,jackson_splitting_1984nourl}.
Vertex splitting has also been studied in the context of \emph{split thickness}, which is 
the minimum maximum number of splits per vertex to obtain a graph with a certain property, e.g., a planar graph or an interval graph~\cite{ekkllm-pstg-18}.

In this paper we consider reducing or removing edge crossings by vertex splitting. Vertex splitting is a useful technique to visualize complex pathway graphs of biological mechanisms. Nielsen et al.~\cite{Nielsen19_bibtex} have proposed an approach using Machine Learning to facilitate the visualization of pathway graphs by training a Support Vector Machine with actions taken during manual biocuration. Henry et al.~\cite{Henry08_bibtex} have improved the readability of clustered social networks using vertex splitting. These studies demonstrate the importance of vertex splitting for visualizing complex networks.

We study variations of the algorithmic problem of constructing planar or crossing-minimal 2-layer drawings with vertex splitting.
In visualizing graphs defined on anatomical structures and cell types in the human body~\cite{github}, 
the two vertex sets of $G$ play different roles and vertex splitting is permitted only on one side of the layout. 
This motivates 
our interest in splitting only the bottom vertices.
The top vertices may either be specified with a given context-dependent input ordering, e.g., alphabetically, following a hierarchy structure, or sorted according to an important measure, or we may be allowed to arbitrarily permute them to perform fewer vertex splits. 

\subsection{Contributions}
We prove that for a given integer $k$ it is \NP-complete to decide whether $G$ admits a planar 2-layer drawing with an arbitrary permutation on the top layer and at most $k$ vertex splits on the bottom layer (see \cref{th:crs-hard}). 
\NP-completeness also holds if at most $k$ vertices can be split, but each an arbitrary number of times (see \cref{th:crsv-hard}).

If, however, the vertex order of $V_t$ is given, then we present two linear-time algorithms to compute planar 2-layer drawings, one minimizing the total number of splits (see \cref{th:crs-fixed-algo}), and one minimizing the number of split vertices (see \cref{th:crsv-fixed-algo}). In view of their linear-time complexity, our algorithms may be useful for practical applications; we perform an experimental evaluation of the algorithm for \cref{th:crs-fixed-algo} using real-world data sets stemming from anatomical structures and cell types in the human body~\cite{github}.

We further study the setting in which the goal is to minimize the number of crossings (but not necessarily remove all of them) using a prescribed total number of splits. For this setting, we prove \NP-completeness even if the vertex order of $V_t$ is given (see \cref{th:crms-fixed-hard}). On the other hand, we provide an \XP-time algorithm parameterized by the number of allowed splits (see \cref{th:crms-fixed-algo}), which, in other words, means that the algorithm has a polynomial running time for any fixed number of allowed splits.

\section{Preliminaries}

We denote the order of the vertices in $V_t$ and $V_b$ in a $2$-layer drawing by $\pi_t$ and $\pi_b$, resp. If a vertex $u$ precedes a vertex $v$, then we denote it by $u \prec v$. Although $2$-layer drawings are defined geometrically, their crossings are fully described by $\pi_t$ and $\pi_b$, as in the following folklore lemma.

\begin{lemma}\label{lemma:crossing_free_position}
Let $\Gamma$ be a 2-layer drawing of a bipartite graph $G=(V_t \cup V_b, E)$. Let $(v_1,u_1)$ and $(v_2,u_2)$ be two edges of $E$ such that $v_1 \prec v_2$ in $\pi_t$. Then, edges $(v_1,u_1)$ and $(v_2,u_2)$ cross each other in $\Gamma$ if and only if $u_2  \prec u_1$ in $\pi_b$.
\end{lemma}

\noindent In the following we formally define the problems we study. 
For all of them, the input contains a bipartite graph $G=(V_t \cup V_b, E)$ and a split parameter $k$.


\begin{itemize}
\item\textbf{Crossing Removal with $\boldsymbol{k}$ Splits  -- \crs{\boldsymbol k}}: Decide if there is a planar 2-layer drawing of $G$ after applying at most $k$ vertex-splits to the vertices~in~$V_b$.
\item\textbf{Crossing Removal with $\boldsymbol{k}$ Split Vertices  -- \crsv{\boldsymbol k}}: Decide if there~is a planar 2-layer drawing of $G$ after splitting  at most $k$ original vertices of $V_b$.
\item\textbf{Crossing Minimization with $\boldsymbol{k}$ Splits  -- \cms{\boldsymbol k, \boldsymbol M}}: Decide if there is a 2-layer drawing of $G$ with at most $M$ crossings after applying at most $k$ vertex-splits to the vertices~in~$V_b$, where $M$ is an additional integer specified as part of the input.
\end{itemize}

Note that in \crsv{k}, once we decide to split an original vertex, then we can further split its copies without incurring any additional cost. The example in \cref{fig:CRSvsCRSVa}--\ref{fig:CRSvsCRSVc} demonstrates the difference between CRS and CRSV.

For all problems, we refer to the variant where the order $\pi_t$ of the vertices in $V_t$ is given as part of the input by adding the suffix \emph{``with Fixed Order''}.
%

\begin{figure*}
    \centering
    \begin{subfigure}[c]{0.32\textwidth}
    \centering
    \includegraphics[page=1]{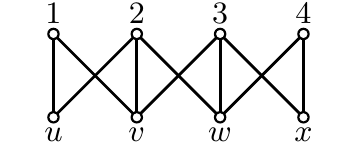}
    \caption{Instance $G$.}
    \label{fig:CRSvsCRSVa}
    \end{subfigure}\hfill
    \begin{subfigure}[c]{0.32\textwidth}
    \centering
    \includegraphics[page=2]{optimalsolution}
    \caption{Optimal CRS solution of $G$ with three splits of three different vertices.}
    \label{fig:CRSvsCRSVb}
    \end{subfigure}\hfill
    \begin{subfigure}[c]{0.32\textwidth}
    \centering
    \includegraphics[page=3]{optimalsolution}
    \caption{Optimal CRSV solution of $G$ with two split vertices.}
    \label{fig:CRSvsCRSVc}
    \end{subfigure}
    \begin{subfigure}[c]{0.32\textwidth}
    \centering
    \includegraphics[page=1]{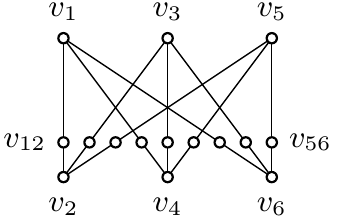}
    \caption{Subdivided graph $G^\prime$.}
    \label{fig:hamiltoniana}
    \end{subfigure}\hfill
    \begin{subfigure}[c]{0.32\textwidth}
    \centering
    \includegraphics[page=2]{hamiltonian}
    \caption{Instance of \crs{k} for $G'$.}
    \label{fig:hamiltonianb}
    \end{subfigure}\hfill
    \begin{subfigure}[c]{0.32\textwidth}
    \centering
    \includegraphics[page=3]{hamiltonian}
    \caption{Splits are minimized if and only if $G$ has a Hamiltonian path.}
    \label{fig:hamiltonianc}
    \end{subfigure}
    
    \begin{subfigure}[c]{0.32\textwidth}
    \centering
    \includegraphics[page=2, width=.9\textwidth]{figures/ipe/algorithm.pdf}
    \caption{Case 1.}
    \label{fig:CRSOpta}
    \end{subfigure}\hfill
    \begin{subfigure}[c]{0.32\textwidth}
    \centering
    \includegraphics[page=3, width=.9\textwidth]{figures/ipe/algorithm.pdf}
    \caption{Case 2.}
    \label{fig:CRSOptb}
    \end{subfigure}\hfill
    \begin{subfigure}[c]{0.32\textwidth}
    \centering
    \includegraphics[page=4, width=.9\textwidth]{figures/ipe/algorithm.pdf}
    \caption{Case 3.}
    \label{fig:CRSOptc}
    \end{subfigure}
    
    \caption{Differences between CRS and CRSV problems (a)--(c); illustrations for the reduction in \cref{th:crs-hard} (d)--(f); illustrations for the optimization algorithm for \crs{k} with Fixed Order, where vertices in $N^+$ are colored in shades of gray (g)--(i).
    }
\end{figure*}

The following lemma implies conditions under which a vertex split must occur.

\begin{lemma}\label{le:split-condition}
Let $G=(V_t \cup V_b, E)$ be a bipartite graph and let $u \in V_b$ be a bottom vertex adjacent to two top vertices $v_1,v_2 \in V_t$, with $v_1 \prec v_2$ in $\pi_t$. In any planar 2-layer drawing of $G$ in which $u$ is not split, we have that:
\begin{enumerate}[C.1]
    \item\label{cond-split:between} 
    A top vertex that appears between $v_1$ and $v_2$ in $\pi_t$ can only be adjacent to $u$;
    \item\label{cond-split:first-last} In $\pi_b$, $u$ is the last neighbor of $v_1$ and the first neighbor of $v_2$.
\end{enumerate}
\end{lemma}
\begin{proof}
If there is a top vertex $v'$ between $v_1$ and $v_2$ adjacent to a bottom vertex $u' \neq u$, then $(v',u')$ crosses $(v_1,u)$ or $(v_2,u)$.
If there is a neighbor $u''$ of $v_1$ after $u$ in $\pi_b$, then the edges $(v_1,u'')$ and $(v_2,u)$ cross. A symmetric argument holds when there is a neighbor of $v_2$ before $u$ in $\pi_b$.
\end{proof}

\section{Crossing Removal with \texorpdfstring{$\boldsymbol{k}$}{Lg} Splits}
\label{se:min-splits}

In this section, we prove that the \crs{k} problem is \NP-complete in general and linear-time solvable when the order $\pi_t$ of the top vertices is part of the input.

\begin{theorem}\label{th:crs-hard}
The \crs{k} problem is \NP-complete.
\end{theorem}
\begin{proof}
The problem belongs to \NP\ since, given a set of at most $k$ splits for the vertices in $V_b$, we can check whether the resulting graph is  planar $2$-layer~\cite{emw-oacp-86}. 


We use a reduction from the \emph{Hamiltonian Path} problem to show the \NP-hardness. 
%
Given an instance $G=(V,E)$ of the Hamiltonian Path~problem, we denote by $G'$ the bipartite graph obtained by subdividing every edge of $G$ once \textcolor{blue}{(\cref{fig:hamiltoniana})}. We construct an instance of the \crs{k} problem \textcolor{blue}{(\cref{fig:hamiltonianb})} by setting the top vertex set $V_t$ to consist of the original vertices of $G$, the bottom vertex set $V_b$ to consist of the subdivision vertices of $G'$, and the split parameter to $k = |E| - |V| + 1$. The reduction can be easily performed in linear time. We prove the equivalence.

Suppose that $G$ has a Hamiltonian path $v_1, \dots, v_n$.
Set $\pi_t=v_1, \dots, v_n$, and split all the vertices of $V_b$, except for the subdivision vertex of the edge $(v_i,v_{i+1})$, for each $i=1,\dots,n-1$ {\textcolor{blue}(\cref{fig:hamiltonianc})}. This results in $|V_b| - (n - 1)$ splits, which is equal to $k$, since $|V_b|=|E|$ and $n = |V|$.
We then construct $\pi_b$ such that, for each $i=1,\dots,n-1$, all the neighbors of $v_i$ appear before all the neighbors of $v_{i+1}$, with their common neighbor being the last neighbor of $v_i$ and the first of $v_{i+1}$. This guarantees that both conditions of \cref{le:split-condition} are satisfied for every vertex of $V_b$. Together with \cref{lemma:crossing_free_position}, this guarantees that the  
2-layer drawing is planar.

Suppose now that $G'$ admits a planar 2-layer drawing with at most $|E| - |V| + 1$ splits. Since $|E| = |V_b|$ and every vertex of $V_b$ has degree exactly $2$ 
(subdivision vertices), there exist at least $|V| - 1$ vertices in $V_b$ that are not split. Consider any such vertex $u \in V_b$. By  C.\ref{cond-split:between} of \cref{le:split-condition}, the two neighbors of $u$ are consecutive in $\pi_t$. Also, these vertices are connected in $G$ by the edge whose subdivision vertex is $u$. Since this holds for each of the at least $|V| - 1$ non-split vertices, we have that each of the $|V| - 1$ distinct pairs of consecutive vertices in $V_t$ (recall that $V_t = V$) is connected by an edge in $G$. Thus, $G$ has a Hamiltonian path.
\end{proof}

Next, we focus on the optimization version of the \crs{k} with the Fixed Order problem.
Our recursive algorithm considers a constrained version of the problem, where the first neighbor in $\pi_b$ of the first vertex in $\pi_t$ may be prescribed. 
At the outset of the recursion, there exists no prescribed first neighbor. The algorithm returns the split vertices in $V_b$ and the corresponding order $\pi_b$.

In the base case, there is only one top vertex $v$, i.e., $|V_t|=1$. Since all vertices in $V_b$ have degree $1$, no split takes place. We set $\pi_b$ to be any order of the vertices in $V_b$ where the first vertex is the prescribed first neighbor of $v$, if~any.

In the recursive case when $|V_t|>1$, we label the vertices in $V_t$ as $v_1,\dots,v_{|V_t|}$, according to $\pi_t$. If the first neighbor of $v_1$ is prescribed, we denote it by $u_1^*$.
Also, we denote by $N^1$ the set of degree-$1$ neighbors of $v_1$, and by $N^+$ the other neighbors of $v_1$.
Note that only the vertices in $N^+$ are candidates to be split for~$v_1$. In particular, by  C.\ref{cond-split:between} of  \cref{le:split-condition}, a vertex in $N^+$ can avoid being split only if it is also incident to $v_2$. Further, since any vertex in $N^+$ that is not split must be the last neighbor of $v_1$ in $\pi_b$, by C.\ref{cond-split:first-last} of \cref{le:split-condition}, at most one of the common neighbors of $v_1$ and $v_2$ will not be split. Analogously, if $u_1^*$ is prescribed, then it must be split, unless $v_1$ has degree $1$.

In view of these properties, we distinguish three cases based on the common neighborhood of $v_1$ and $v_2$. In all cases, we will recursively compute a solution for the instance composed of the graph $G'=(V_t' \cup V_b',E')$ obtained by removing $v_1$ and the vertices in $N^1$ from $G$, and of the order $\pi_t' = v_2, \dots, v_{|V_t|}$. We denote by $\pi_b'$ and $s'$ the computed order and the number of splits for the vertices in $V_b'$.
In the following we specify for each case whether the first neighbor of $v_2$ in the new instance is prescribed or not, and how to incorporate the neighbors of $v_1$ into $\pi_b'$.

\textbf{Case 1: $v_1$ and $v_2$ have no common neighbor}; 
see \cref{fig:CRSOpta}.
In this case, we do not prescribe the first neighbor of $v_2$ in the instance composed of $G'$ and $\pi_t'$.
To compute a solution for the original instance, we split each vertex in $N^+$ so that one copy becomes incident only to $v_1$. 
We construct $\pi_b$ by selecting the prescribed vertex $u_1^*$, if any, followed by the remaining neighbors of $v_1$ in any order and, finally, by appending $\pi_b'$. This results in $s = |N^+| + s'$ splits.

\textbf{Case 2: $v_1$ and $v_2$ have exactly one common neighbor $u$.}
If $u = u_1^*$ and $v_1$ have a degree larger than $1$, then $u$ cannot be the last neighbor of $v_1$ and must be split. Thus, we perform the same procedure as in Case 1. 
Otherwise, in the instance composed of $G'$ and $\pi_t'$, we set $u$ as the prescribed first neighbor of $v_2$; 
{blue}see \cref{fig:CRSOptb}.
To compute a solution for the original instance, we split each vertex in $N^+$, except $u$, so that one copy becomes incident only to $v_1$. 
We construct $\pi_b$ by selecting the prescribed vertex $u_1^*$, if any, followed by the remaining neighbors of $v_1$ different from $u$ in any order and, finally, by appending $\pi_b'$. 
This results in $s = |N^+| -1 + s'$ splits.

\textbf{Case 3: $v_1$ and $v_2$ have more than one common neighbor.}
If $v_1$ and $v_2$ have exactly two common neighbors $u,u'$ and one of them is $u_1^*$, say $u = u_1^*$, then $u$ cannot be the last neighbor of $v_1$, as $v_1$ has degree larger than $1$. Thus, we proceed exactly as in Case 2, using $u'$ as the only common neighbor of $v_1$ and $v_2$.

Otherwise, there are at least two neighbors of $v_1$ different from $u_1^*$; 
see \cref{fig:CRSOptc}.
We want to choose one of these vertices as the last neighbor of $v_1$, so that it is not split. However, the choice is not arbitrary 
as this may affect the possibility for $v_2$ to save the split for a neighbor it shares with $v_3$. 
In the instance composed of $G'$ and $\pi_t'$, we do not prescribe the first vertex of $v_2$.
To compute a solution for the original instance, we simply choose as the last neighbor of $v_1$ any of its common neighbors with $v_2$ that has not been set as the last neighbor of $v_2$ in $\pi_b'$. 
Such a vertex, say $u$, always exists since $v_1$ and $v_2$ have at least two common neighbors different from $u_1^*$, and can be moved to become the first vertex in $\pi_b'$.
%
Specifically, we split all the vertices in $N^+$, except for $u$, so that one copy becomes incident only to $v_1$. 
We construct $\pi_b$ by selecting the prescribed vertex $u_1^*$, if any, followed by the remaining neighbors of $v_1$ different from $u$ in any order. 
We then modify $\pi_b'$ by moving $u$ to be the first vertex. Note that this operation does not affect planarity, as it only involves reordering the set of consecutive degree-$1$ vertices incident to $v_2$.
Finally, we append the modified $\pi_b'$. This results in $s = |N^+| - 1 + s'$ splits.

\begin{theorem}\label{th:crs-fixed-algo}
For a bipartite graph $G=(V_t \cup V_b,E)$ and an order $\pi_t$ of $V_t$, the optimization version of \crs{k} with Fixed Order can be solved in $O(|E|)$ time.
\end{theorem}
\begin{proof}
By construction, for each $i = 1, \dots, |V_t|-1$, all neighbors of $v_i$ precede all neighbors of $v_{i+1}$ in $\pi_b$. Thus, by \cref{lemma:crossing_free_position}, the drawing is planar.
The minimality of the number of splits follows from  \cref{le:split-condition}, as discussed before the case distinction. In particular, any minimum-splits solution can be shown to be equivalent to the one produced by our algorithm.
The time complexity follows as each vertex only needs to check its neighbors a constant number of times.
\end{proof}

\begin{algorithm*}
    \caption{\crs{k}-with-Fixed-Order($V_t, V_b, \pi_t$)}
    \begin{algorithmic}
    \If{there is only one vertex in $V_t$}
        \State Let $N(v)$ be the neighbor vertices of $v$
        \State Output an arbitrary order $N(v)$
    \Else
        \State Let $v_1$ and $v_2$ are the first two vertices in $V_t$
        \If{$N(v_1) \cap N(v_2) \neq \emptyset$}
            \State Let $N^+(v_1)$ be the set of neighbors of $v_1$ having degree more than one
            \State let $\pi_b'$ be the output of \crs{k}-with-Fixed-Order($V_t-v_1, V_b, \pi_t-v_1$)
            \State Find a common neighbor $u$ of $v_1$ and $v_2$ not in $\pi_b'$
            \State output $u, \pi_b'$
        \Else
            \State let $\pi_b'$ be the output of \crs{k}-with-Fixed-Order($V_t-v_1, V_b, \pi_t-v_1$)
            \State Output $N(v_1), \pi_b'$
        \EndIf 
    \EndIf
    \end{algorithmic}
\end{algorithm*}

We conclude this section by mentioning that the \crs{k} problem had already been considered, under a different terminology, in the context of molecular QCA circuits design~\cite{DBLP:journals/tcad/ChaudharyCHNRW07}. Here, the problem was claimed to be \NP-complete, without providing a formal proof. In the same work, when the order $\pi_t$ of the top vertices is part of the input, an alternative algorithm was proposed based on the construction of an auxiliary graph that has super-linear size. Exploiting linear-time sorting algorithms and observations that allow avoiding explicitly constructing all edges of this graph, the authors were able to obtain a linear-time implementation.
We believe that our algorithm of \cref{th:crs-fixed-algo} is simpler and more intuitive, and directly leads to a linear-time implementation.

\section{Crossing Removal with \texorpdfstring{$\boldsymbol{k}$}{Lg} Split Vertices}

In this section, we prove that the \crsv{k} problem is \NP-complete in general and linear-time solvable when the order $\pi_t$ of the top vertices is part of the input.
Ahmed et al.~\cite{AhmedKK22} showed that \crsv{k} is \FPT\  when parameterized by $k$.
To prove the \NP-completeness we can use the reduction of  \cref{th:crs-hard}. In fact, in the graphs produced by that reduction all vertices in $V_b$ have degree $2$. Hence, the number of vertices that are split coincides with the total number of splits.

\begin{theorem}\label{th:crsv-hard}
The \crsv{k} problem is \NP-complete.
\end{theorem}

For the version with Fixed Order, we first use C.\ref{cond-split:between} of  \cref{le:split-condition} to identify vertices that need to be split at least once, and repeatedly split them until each has degree $1$.
For a vertex $u \in V_b$, we can decide if it needs to be split by checking whether its neighbors are consecutive in $\pi_t$ and, if $u$ has degree at least $3$, all its neighbors different from the first and last have degree exactly $1$. 

We first perform all necessary splits. 
For each $i = 1, \dots, |V_t|-1$, consider the two consecutive top vertices $v_i$ and $v_{i+1}$. 
If they have no common neighbor, no split is needed. 
%
If they have exactly one common neighbor $u$, then we set $u$ as the last neighbor of $v_i$ and the first of $v_{i+1}$, which allows us not to split $u$, according to C.\ref{cond-split:first-last}.
Since $u$ did not participate in any necessary split, if $u$ is also adjacent to other vertices, then all its neighbors have degree $1$, except possibly the first and last. Hence, C.\ref{cond-split:first-last} 
can be guaranteed for all pairs of consecutive neighbors of $u$.

Otherwise, $v_i$ and $v_{i+1}$ have at least two common neighbors and thus have degree at least $2$. 
Hence, all common neighbors of $v_i$ and $v_{i+1}$ must be split, except for at most one, namely the one that is set as the last neighbor of $v_i$ and as the first of $v_{i+1}$. Since all these vertices are incident only to $v_i$ and $v_{i+1}$, as otherwise they would have been split by C.\ref{cond-split:between}, we can arbitrarily choose any of them, without affecting the splits of other vertices. 

\begin{algorithm*}
    \caption{\crsv{k}-with-Fixed-Order($V_t, V_b, \pi_t$)}
    \begin{algorithmic}
        \For{$i = 1, \cdots, |V_t|-1$}
            \If{$v_i$ and $v_{i+1}$ has no common neighbors}
                \State No split is needed
            \EndIf
            \If{$v_i$ and $v_{i+1}$ has exactly one common neighbor}
                \State Let $u$ be the common neighbor
                \If{$u$ is also adjacent to other vertices that has degree 1 except the first an last one}
                    \State No split is needed
                \Else
                    \State Split $u$ when mandatory
                \EndIf
            \Else
                \State Split all common neighbors except at most once
            \EndIf
        \EndFor
    \end{algorithmic}
\end{algorithm*}

At the end 
we construct the order $\pi_b$  
so that, for each $i = 1, \dots, |V_t|-1$, all the neighbors of $v_i$ precede all the neighbors of $v_{i+1}$, and the unique common neighbor of $v_i$ and $v_{i+1}$, if any, is the last neighbor of $v_i$ and the first of $v_{i+1}$. By \cref{lemma:crossing_free_position}, this guarantees planarity. Identifying and performing all unavoidable splits and computing $\pi_b$ can be easily done in $O(|E|)$ time. Since we only performed unavoidable splits, as dictated by \cref{le:split-condition}, we have the following.

\begin{theorem}\label{th:crsv-fixed-algo}
For a bipartite graph $G=(V_t \cup V_b,E)$ and an order $\pi_t$ of $V_t$, the optimization version of \crsv{k} with Fixed Order minimizing the number of split vertices can be solved in $O(|E|)$ time.
\end{theorem}

\section{Crossing Minimization with \texorpdfstring{$\boldsymbol{k}$}{Lg} Splits}

In this section we consider minimizing crossings (not necessarily removing all), by applying at most $k$ splits. We first prove \NP-completeness of the decision problem \cms{k,M} with Fixed Order and then give a polynomial-time algorithm assuming the integer $k$ is a constant.

\begin{theorem}\label{th:crms-fixed-hard}
For a bipartite graph $G=(V_t \cup V_b,E)$, an order $\pi_t$ of $V_t$, and integers $k,M$, problem \cms{k,M} with Fixed Order is \NP-complete.
\end{theorem}

\begin{proof}
We reduce from the \NP-complete \textsc{Decision Crossing Problem} (\textsc{DCP})~\cite{EadesW94}, where given a bipartite 2-layer graph with one vertex order fixed, the goal is to find an order of the other set such that the number of crossings is at most a given integer $M$. 
Given an instance of \textsc{DCP}, i.e.,a 2-layer graph $G=(V_t\cup V_b,E)$, with ordering $\pi_t$ of $V_t$ and integer $M$, we construct an instance $G'$ of \cms{k,M} where $k=|V_b|$. First let $G'=(V'_t \cup V'_b,E')$ be a copy of $G$. We give an arbitrary ordering $\pi_b$ to the vertices of $V'_b$. 
We then  add, respectively,  to each vertex set $V_t$ and $V_b$ a set $U_t$ and $U_b$ of $M+1$ vertices and connect each $u \in U_t$ to exactly one $v \in U_b$, forming a matching of size $M+1$. 
We add the vertices of $U_t$ to $\pi_t$ (resp. $U_b$ to $\pi_b$) after all the vertices of $V_t$ ($V_b)$. 
We lastly add a set $W_t$ of $k$ vertices to $V'_t$, placed at the end of $\pi_t$, such that each $w_i \in W_t$ ($i=1, \dots, k$) has exactly one neighbor $v_i\in V_b$ and vice versa; 
see \cref{fig:red_splita}.

Given an ordering $\pi_b^*$ of $V_b$ that results in a drawing of $G$ with at most $M$ crossings, we show that we can solve the \cms{k,M} instance $G'$. In $G'$, we split each vertex of $V_b$ to obtain the sets $V_b^1$ and $V_b^2$ in which we place exactly one copy of each original vertex. 
We place $V_b^2$ after the vertices of $U_b$ in $\pi_b$ in the same order that the vertices of $W_t$ appear in $\pi_t$ and draw a single edge between the copies and their neighbor in $W_t$. 
We place $V_b^1$ before the vertices of $M_2$ in $\pi_b$ in the same order as in $\pi_b^*$.
The graph induced by $V_b^1$ and $V_t$ is the same graph as $G$, hence it has at most $M$ crossings. Since $V_t$ only has neighbors in $V_b$ and all those neighbors are in $V_b^1$, it has no other outgoing edges, similarly, all edges incident to vertices in $W_t$ are assigned to the copies in $V_b^2$. The remaining graph is crossing-free as the vertices in $U_t$ and $W_t$ form a crossing-free matching with the vertices in $U_b$ and $V_b^2$. 

Conversely, let $G^*$ be a 2-layer drawing obtained from $G'$ after $k$ split operations that has at most $M$ crossings. 
Since each vertex $v\in V_b$ has a neighbor $w\in W_t$, it induces $M+1$ crossings with edges induced by the vertices in $U_t \cup U_b$. 
Since the vertices in $U_b$ have a single neighbor, they cannot be split, thus every vertex in $V_b$ is split once, and their neighborhood are partitioned for each copy in the following way: one copy receives the neighbor in $W_t$ and one copy receives the remaining neighbors, which are in $V_t$ 
(\cref{fig:red_splitb}), thus avoiding the at least $M+1$ crossings induced by $U_t \cup U_b$. Any other split would imply at least $M+1$ crossings.
The graph induced by the copies that receive the neighbors in $V_t$ has at most $M$ crossings, thus, the ordering found for those copies is a solution to the \textsc{DCP} instance $G$.
\end{proof}

\begin{figure*}
    \centering
    \begin{subfigure}[t]{0.47\textwidth}
    \centering
    \includegraphics[page=1, width=.9\textwidth]{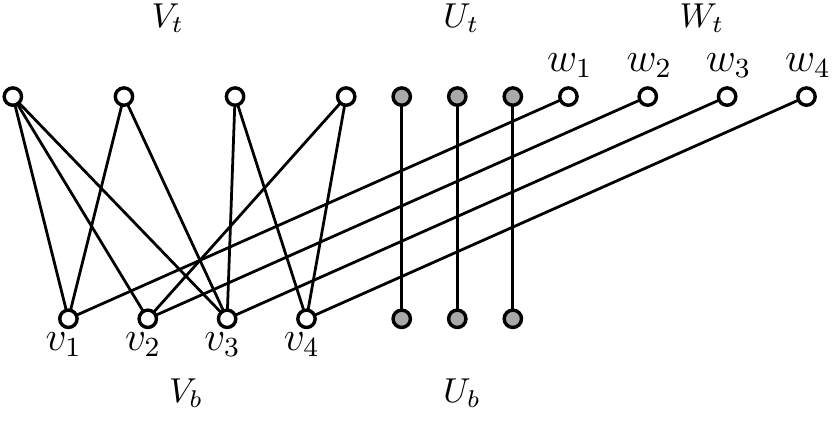}
    \caption{Instance of \cms{k} constructed from a \textsc{DCP} instance, in light gray the vertices in $U_t \cup U_b$;  before splitting.}
    \label{fig:red_splita}
    \end{subfigure}\hfill
    \begin{subfigure}[t]{0.47\textwidth}
    \centering
    \includegraphics[page=2, width=.9\textwidth]{ipe/min_red.pdf}
    \caption{After splitting.}
    \label{fig:red_splitb}
    \end{subfigure}
    \begin{subfigure}[c]{0.99\textwidth}
    \centering
    \includegraphics{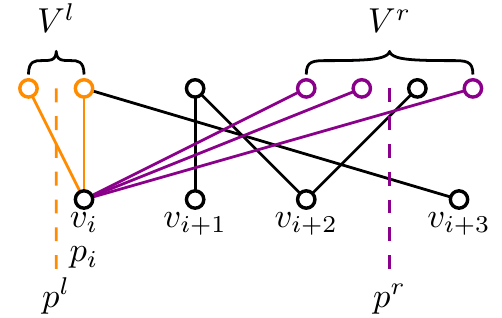}
    \caption{Vertex $v_i$ with span 6 and a split in $V^l$ (span 1) and $V^r$ (span 3). The barycenter of $V^r$ is $p^r$. When moving right from $p_i$ in \textit{CR-count} we process $v_{i+1}$ (reduces $|V^r|$ crossings) and $v_{i+2}$ (reduces $|V^r|$ but adds 1 crossing), but not $v_{i+3}$.}
    \label{fig:heuristics_expl}
    \end{subfigure}
    
    \caption{Illustrations for \cref{th:crms-fixed-hard} (a)--(b); illustration for the crossing reduction heuristics (c).
    }
\end{figure*}

Next, we present a simple \XP-time algorithm for the crossing minimization version of \cms, parameterized by the number $k$ of splits, i.e., the algorithm runs in polynomial time $O(n^{f(k)})$, where $n$ is the input size, $k$ is the parameter, and $f$ is a computable function. 
Let $G=(V_t \cup V_b, E)$ be a 2-layer graph with vertex orders $\pi_t$ and $\pi_b$ and let $k$ be the desired number of splits. 
Our algorithm executes the following steps. 
First, it determines a set of splits by choosing $k$ times a vertex from the $n$ vertices in $V_b$ -- we enumerate all options. 
For any vertex $v\in V_b$ split $i$ times in the first step, $v$ is replaced by the set of copies $\{v_1,...,v_{i+1}\}$. 
The neighborhood $N(v)$ of a vertex $v \in V_b$ is a subset of $V_t$ ordered by $\pi_t$. 
We partition this ordered neighborhood into $i+1$ consecutive subsets, i.e., for each subset, all its elements are sequential in $N(v)$ -- again, we enumerate all possible partitions. 
Each set is assigned to be the neighborhood of one of the copies of $v$.
The algorithm then chooses an ordering of all copies of all split vertices and attempts all their possible placements by merging them into the order $\pi_b$ of the unsplit vertices of $V_b$.
The crossing number of every resulting layout is computed and the graph with minimum crossing number yields the solution to our input. It remains to show that the running time of this algorithm is polynomial for constant $k$.

\begin{theorem}\label{th:crms-fixed-algo}
For a 2-layer graph $G=(V_t \cup V_b, E)$ with vertex orders $\pi_t, \pi_b$ and a constant $k \in \mathbb{N}$ we can minimize the number of crossings by applying at most $k$ splits in time $O(n^{4k})$.
\end{theorem}

\begin{proof}
Let $G^*$ be a crossing-minimal solution after applying $k$ splits on $V_b$ and let us assume that our algorithm would not find a solution with this number of crossings. 
As our algorithm considers all possibilities to apply $k$ splits, it also attempts the splits applied in $G^*$. Similarly, the neighborhood partition of $G^*$ and the copy placement are explicitly considered by the algorithm as it enumerates all possibilities. Hence a solution at least as good as $G^*$ is found, proving correctness.

Let $n_t = |V_t|$ and $n_b = |V_b|$ with $n = n_t + n_b$. The algorithm initially chooses $k$ times from $n_b$ vertices leading to $n_b^k$ possible sets of copies. 
Since a vertex has degree at most $n_t$, there are at most $n_t^k$ possible neighborhoods for each copy. 
Additionally, there are $(2k)!$ orderings of at most $2k$ copies. 
Lastly, there are $n_b^{2k}$ possible placement of the $2k$ ordered copies between the at most $n_b$ unsplit vertices in $\pi_b$.
This leads to an overall runtime of $O((2k)! \cdot n^{4k}) = O(n^{4k})$ to iterate through all possible solutions and select the one with a minimum number of crossings.
\end{proof}

\section{Crossing Reduction Heuristics}

In this section we present two greedy heuristics to iteratively reduce crossings in a two-layer drawing by selecting and splitting vertices; see Algorithms 3--4 in the supplemental material for the pseudocode. The input to the algorithm is a bipartite graph $G = (V_t \cup V_b, E)$, order $\pi_t$ of $V_t$ and order $\pi_b$ of $V_b$. Here we use the barycenter heuristic for the initial orders, but any initial order computed by a crossing-reduction algorithm can also be used.  Additionally, an input parameter $k$ is specified that represents a budget of available split operations. For both heuristics we iteratively perform $k$ splits by selecting the most promising vertex in $V_b$ and a partition of its respective neighbors in $V_t$ into a consecutive set of left neighbors $V^l$ and a consecutive set of right neighbors $V^r$. After splitting, the original vertex receives the set $V^l$ assigned as neighbors and the copy receives the set $V^r$. Next we describe how the vertex to be split is selected in each of the heuristics.

Firstly, in the case of the \textit{max-span} heuristic we select the vertex $v$ with the maximum \emph{span}, i.e., the maximum distance between its leftmost neighbor and rightmost neighbor in $\pi_t$. Then, we process $v$ by iterating in order over its neighbors and assigning them to either the set of left neighbors $V^l$ or right neighbors $V^r$ depending on the index; 
see Fig.~\ref{fig:heuristics_expl}.
In each iteration we compute the sum of the squared span of $V^l$ and $V^r$. The minimum value indicates the best partition of $v$'s neighbors. The complexity of the heuristic is linear in the number of edges $O(|E|)$.  

Secondly, the \textit{CR-count} heuristic selects promising split vertices by computing crossings that can be reduced and selects the vertex with potentially most reduced crossings. 
First, we assign each vertex $v_i$ in $V_b$ a position $p_i \in \mathbb{R}$ which is the barycenter of positions in $\pi_t$ of its neighbors in $V_t$. 
Similarly to \textit{max-span}, we iterate in order over the neighbors of all $v_i$ creating partitions $V^l$ and $V^r$. In each iteration we compute the barycenter $p^l$ of $V^l$ and $p^r$ of $V^r$; 
see \cref{fig:heuristics_expl}.
Next, we start from position $p_i$ and move in ascending order processing all other vertices in $V_b$ until we reach $p^r$. 
To process a vertex $v_j \in V_b$ we look at all edges to its neighbors $N(v_j) \subseteq V_t$ and use a case distinction to count how many crossings can be reduced. In the first case a neighbor in $V_t$ is left of the leftmost vertex in $V^r$. Here, we would reduce $|V^r|$ crossings as the new position of $v_r$ would not cross. In the second case a neighbor is between the leftmost and rightmost vertex in $V^r$, i.e., we would reduce some crossings but add others. In the third case a neighbor is right of the rightmost vertex in $V^r$, i.e., no crossings can be reduced. Likewise, we process vertices to the left of $v_i$ up to the barycenter $p^l$. Finally, after computing all potentially reducible crossings for each vertex and split combination, we select the combination reducing the crossings most, assign $V_l$ to the original vertex and $V_r$ to the new copy. Both vertices are positioned in their respective barycenters $p^l$ and $p^r$ in the order. The complexity of the algorithm is $O(|V_b|^2 |V_t|)$.

\newpage

\begin{algorithm}[H]
\caption{max-span heuristic}\label{alg:maxspan}
    \begin{algorithmic}
        \Require $G=(V_t \cup V_b, E)$, $\pi_t$, $\pi_b$, $k$
        \Ensure $G^\prime$, $V^\prime_b$, $\pi_b^\prime$

        \State $G^\prime \gets G$
        \State $V^\prime_b \gets V_b$
        \State $\pi_b^\prime \gets \pi_b$

        \While{$k > 0$}
            \State $maxSpan \gets 0$
            \State $splitNode \gets None$
        
            \For{$i = 1, \cdots, |V^\prime_b|$}
                \State \Comment{compute index of leftmost/rightmost neighbor of $v_i$ in $\pi_t$} 
                \State $minNeighbor \gets \text{leftmostNeighbor}(v_i, N(v_i), \pi_t)$ 
                \State $maxNeighbor \gets \text{rightmostNeighbor}(v_i, N(v_i), \pi_t)$  
                \State $span \gets maxNeighbor - minNeighbor$

                \If{$span > maxSpan$}
                    \State $maxSpan \gets span$
                    \State $splitNode \gets v_i$
                \EndIf
            \EndFor

            \If{$maxSpan == 0$} \Comment{no split possible}
                \State \Return $G, V^\prime_b, \pi^\prime_b$
            \Else
                \State $N^o \gets orderByIndex(N(splitNode), \pi_t)$
                \State $V^l \gets \emptyset$
                \State $V^r \gets N^o$
                \State $minSpan \gets (\text{computeSpan}(V^r, \pi_t))^2$
                \State $V^l_{s} \gets  V^l$
                \State $V^r_{s} \gets  V^r$
                \For{$i=1,\dots,|N^o|$}
                    \State $V^l \gets V^l \cup \{v_i\}$
                    \State $V^r \gets V^r \setminus \{v_i\}$
                    \State $span \gets (\text{computeSpan}(V^l, \pi_t))^2 + (\text{computeSpan}(V^r, \pi_t))^2$ 

                    \If{$span < minSpan$}
                        \State $minSpan \gets span$
                        \State $V^l_{s} \gets  V^l$
                        \State $V^r_{s} \gets  V^r$
                    \EndIf
                \EndFor

                \State \Comment{update $G$ and $\pi_b$ by removing $v_i$ and adding $v_l$ and $v_r$ with neighbors $V^l_s$ and $V^r_s$}
                \State \Comment{$v_l$ and $v_r$ get the placed in the barycenter of $V^l_s$ and $V^r_s$ in the order $\pi^\prime_b$}
                \State $\text{updateGraphAndOrder}(G^\prime, V^\prime_b, \pi^\prime_b, V^l_s, V^r_s)$
            \EndIf
            
            \State $k \gets k - 1$
    
        \EndWhile

        \State \Return $G, V^\prime_b, \pi_b$
    \end{algorithmic}
\end{algorithm}

\begin{algorithm}[H]
\caption{CR-count heuristic}\label{alg:crcount}
    \begin{algorithmic}
        \Require $G=(V_t \cup V_b, E)$, $\pi_t$, $\pi_b$, $k$
        \Ensure $G^\prime$, $V^\prime_b$, $\pi_b^\prime$

        \State $G^\prime \gets G$
        \State $V^\prime_b \gets V_b$
        \State $\pi_b^\prime \gets \pi_b$

        \While{$k > 0$}

            \State $V^l_s \gets \emptyset$
            \State $V^r_s \gets \emptyset$
            \State $bestSplit \gets 0$
                
            \For{$i=1,\dots,|V_b^\prime|$}
                \State $N^o \gets orderByIndex(N(splitNode), \pi_t)$
                \State $V^l \gets \emptyset$
                \State $V^r \gets N^o$
                \State $p_i \gets \text{computeBarycenter}(G, N^o, \pi_t)$

                \For{$j=1,\dots,|N^o|$}
                    \State $V^l \gets V^l \cup \{v_j\}$
                    \State $V^r \gets V^r \setminus \{v_j\}$
                    \State $p^l \gets \text{computeBarycenter}(G, V^l, \pi_t)$
                    \State $p^r \gets \text{computeBarycenter}(G, V^r, \pi_t)$

                    \State \Comment{count avoidable crossings in a range, i.e. between $p_i$ and $p^r$}
                    \State $count^l \gets \text{countCrossingsInRange}(G, \pi_t, \pi^\prime_b, p^l, p_i)$
                    \State $count^r \gets \text{countCrossingsInRange}(G, \pi_t, \pi^\prime_b, p_i, p^r)$
                
                    \If{$bestSplit < count^l + count^r$}
                        \State $bestSplit \gets count^l + count^r$ 
                        \State $V^l_s \gets V^l$
                        \State $V^r_s \gets V^r$
                    \EndIf
                \EndFor
                
            \EndFor

            \If{$bestSplit == 0$} \Comment{no split possible}
                \State \Return $G, V^\prime_b, \pi^\prime_b$
            \Else
                \State \Comment{update $G$ and $\pi_b$ by removing $v_i$ and adding $v_l$ and $v_r$ with neighbors $V^l_s$ and $V^r_s$}
                \State \Comment{$v_l$ and $v_r$ get the placed in the barycenter of $V^l_s$ and $V^r_s$ in the order $\pi^\prime_b$}
                \State $\text{updateGraphAndOrder}(G^\prime, V^\prime_b, \pi^\prime_b, V^l_s, V^r_s)$
            \EndIf
            
            \State $k \gets k - 1$
    
        \EndWhile

        \State \Return $G, V^\prime_b, \pi_b$
    \end{algorithmic}
\end{algorithm}

\section{Experimental Results}
We have experimentally evaluated four of the five algorithms described earlier with 22 real-world datasets: the exact algorithm of \crs{k} with Fixed Order, the exact algorithm of \crsv{k} with Fixed Order, and two  heuristics for crossing reduction. The algorithm behind Thm.~\ref{th:crms-fixed-algo} is inefficient in practice.
We analyze performance w.r.t. the number of crossings in the layouts, number of vertex splits, number of vertices that we split, the maximum number of splits, and running time.
\subsection{Experimental Design}
We have mentioned that 2-layer drawings have been applied in visualizing graphs defined on anatomical structures and cell types in the human body~\cite{github}. There exists a variety of cell types, genes, and proteins related to different organs of the human body. Hierarchical structures have been used to show the relationship between organs to anatomical structures, anatomical structures to cell types, and cell types to genes/proteins. Cell types and genes/proteins situate on a particular layer, unlike anatomical structures. Hence, we can consider a 2-layer graph $G$ where cell types represent one layer and genes/proteins represent another layer and analyze $G$ before and after splitting. In this section, we consider the real-world 2-layer graphs generated from the dataset of different organs and show the experimental results obtained on those graphs. 

\subsection{Datasets}
We use 22 real-world instances of 2-layer graphs from~\cite{github}; see Table~\ref{tab:graph_stat}.

\begin{table}[H]
    \begin{center}
    \begin{tabular}{|c|c|c|c|c|c|c|c|}
    \hline
    \textbf{Organ} & \textbf{$|V|$} & \textbf{$|E|$} & \textbf{Cell types} & \textbf{Genes/proteins} & \textbf{Density} & \textbf{Max degree}\\ \hline
    Blood & 179 & 461 & 30 & 149 & 0.0289 & 57 \\ \hline
    Fallopian Tube & 42 & 32 & 19 & 23 & 0.0371 & 3 \\ \hline
    Lung & 231 & 231 & 69 & 162 & 0.008 & 8 \\ \hline
    Peripheral Nervous System & 3 & 2 & 1 & 2 & 0.666 & 2 \\ \hline
    Thymus & 552 & 658 & 41 & 511 & 0.00432 & 93 \\ \hline
    Heart & 60 & 51 & 15 & 45 & 0.028 & 7 \\ \hline
    Lymph Nodes & 299 & 491 & 44 & 255 & 0.0110 & 36 \\ \hline
    Prostate & 43 & 36 & 12 & 31 & 0.039 & 3 \\ \hline
    Ureter & 44 & 53 & 14 & 30 & 0.0560 & 9 \\ \hline
    Bone Marrow & 343 & 662 & 45 & 298 & 0.011 & 25 \\ \hline
    Kidney & 201 & 237 & 58 & 143 & 0.011 & 8 \\ \hline
    Skin & 102 & 90 & 36 & 66 & 0.017 & 7 \\ \hline
    Urinary Bladder & 46 & 55 & 15 & 31 & 0.053 & 9 \\ \hline
    Brain & 381 & 346 & 127 & 254 & 0.004 & 5 \\ \hline
    Large Intestine & 124 & 139 & 51 & 73 & 0.0182 & 8 \\ \hline
    Ovary & 9 & 6 & 3 & 6 & 0.166 & 2 \\ \hline
    Small Intestine & 18 & 13 & 5 & 13 & 0.084 & 4 \\ \hline
    Uterus & 61 & 65 & 16 & 45 & 0.035 & 9 \\ \hline
    Eye & 145 & 270 & 47 & 98 & 0.0258 & 68 \\ \hline
    Liver & 73 & 57 & 26 & 47 & 0.0216 & 5 \\ \hline
    Pancreas & 69 & 100 & 29 & 40 & 0.042 & 12 \\ \hline
    Spleen & 290 & 414 & 65 & 225 & 0.009 & 23 \\ \hline
    \end{tabular}
    \end{center}
    \caption{Statistics about the organ graphs from the HubMAP dataset~\cite{github}. The density of a graph $G=(V,E)$ with $V=V_t \cup V_b$ is defined as $2|E|/(|V|(|V|-1))$.}
    \label{tab:graph_stat}
\end{table}

\subsection{Interactive Tool Design}
We developed an interactive tool, where the user can upload a dataset and visualize the corresponding 2-layer graph; see Fig.~\ref{fig:interface}. 
A dataset can be loaded by pasting JSON-formatted text in the input area on the left-hand side of the interface.
We use blue and red colors to draw nodes that represent cell types and genes/proteins respectively. There is a legend in the top left corner of the interface to describe the color code. Instead of using top and bottom layers, we use left and right layers, for easier node labeling (i.e., the left and right layers represent $V_t$ and $V_b$). 
There are multiple radio buttons to the configuration of the drawing. 
The user can fix the order of either the blue vertices or the red vertices by selecting one set of radio buttons. 
The number of vertex splits depends on the initial layout. 
We consider two types of initial layouts: the vertices in each layer are positioned in alphabetical order, or in barycentric order~\cite{SugiyamaTT81}. 
The user can select an initial order from the input interface by using another set of radio buttons. There are ``Draw'' and ``Split'' buttons in the interface. Once the user selects an order for the left layer and clicks the draw button, then the initial ordered layout will be shown on the right side of the interface. Clicking the split button replaces the initial layout and shows the final layout on the right side of the interface.

The right output interface is interactive; the user can see further details using different interactions. 
When the graph is large the user can scroll up and down to see different parts of the layout. 
The user can highlight the adjacent edges by clicking on a particular vertex in case of dense layouts. 
We keep the label texts less than or equal to ten characters. 
If a label is longer then we show the first ten characters and truncate the rest. 
If the user puts the mouse over the label or the corresponding vertex, a pop-up message will show the full label. 
If the user moves out the mouse, the message will be removed too. 
Besides showing the full label, we also provide other useful information, e.g., the degree and ID of the vertex; see Figure~\ref{fig:interaction}.

\begin{figure}[H]
    \centering
    \begin{subfigure}[c]{.39\textwidth}
    \centering
    \includegraphics[page=2, width=\textwidth]{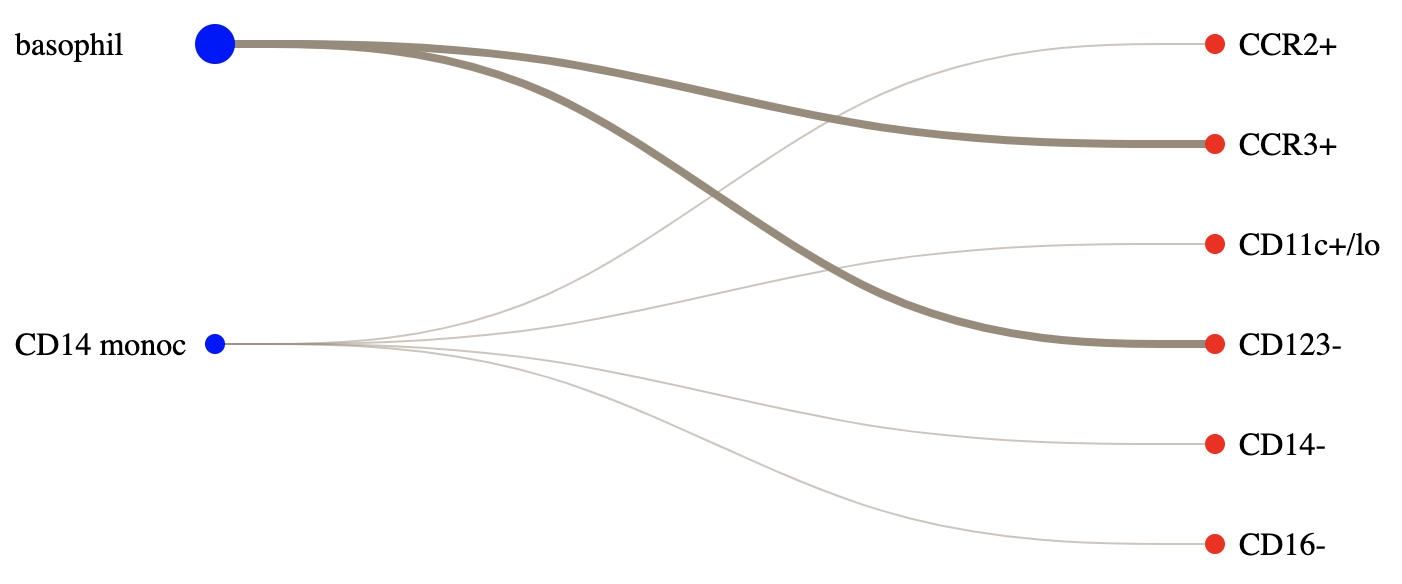}
    \caption{The system highlights the adjacent edges when the user clicks on a vertex (``basophil'' in this case).}
    \end{subfigure}\hfill
    \begin{subfigure}[c]{.39\textwidth}
    \centering
    \includegraphics[page=3, width=\textwidth]{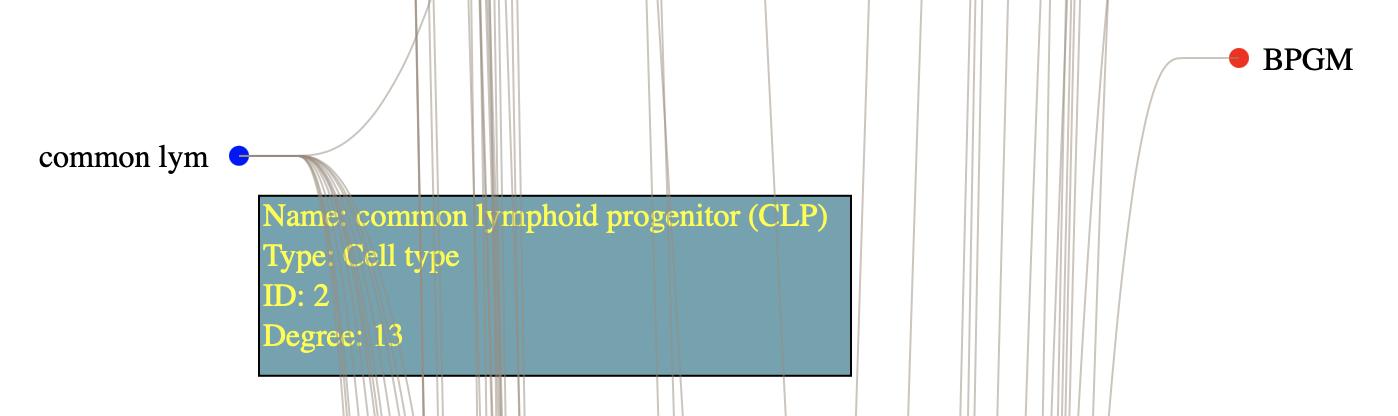}
    \caption{A pop-up message showing the full label, and other related information.}
    \end{subfigure}
    \caption{Interacting with the system.}
    \label{fig:interaction}
\end{figure}

\begin{figure}[H]
  \begin{minipage}[c]{0.48\textwidth}
    \includegraphics[width=1.0\textwidth]{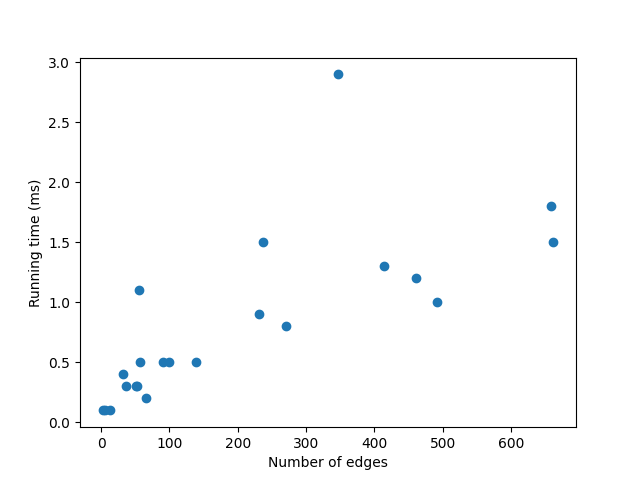}
  \end{minipage}\hfill
  \begin{minipage}[c]{0.48\textwidth}
    \caption{Showing the running time w.r.t. the number of edges. The vertices $V_t$ are the cell types and the vertices $V_b$ are the genes/proteins. The vertices are alphabetically ordered. This is corresponding to the exact algorithm of \crs{k} with Fixed Order. We can see as the number of edges increase the running time increases as expected.}
  \end{minipage}
    \label{fig:edges_time}
\end{figure}

\begin{figure}[H]
    \centering
    \begin{subfigure}[b]{0.70\textwidth}
    \includegraphics[width=\textwidth]{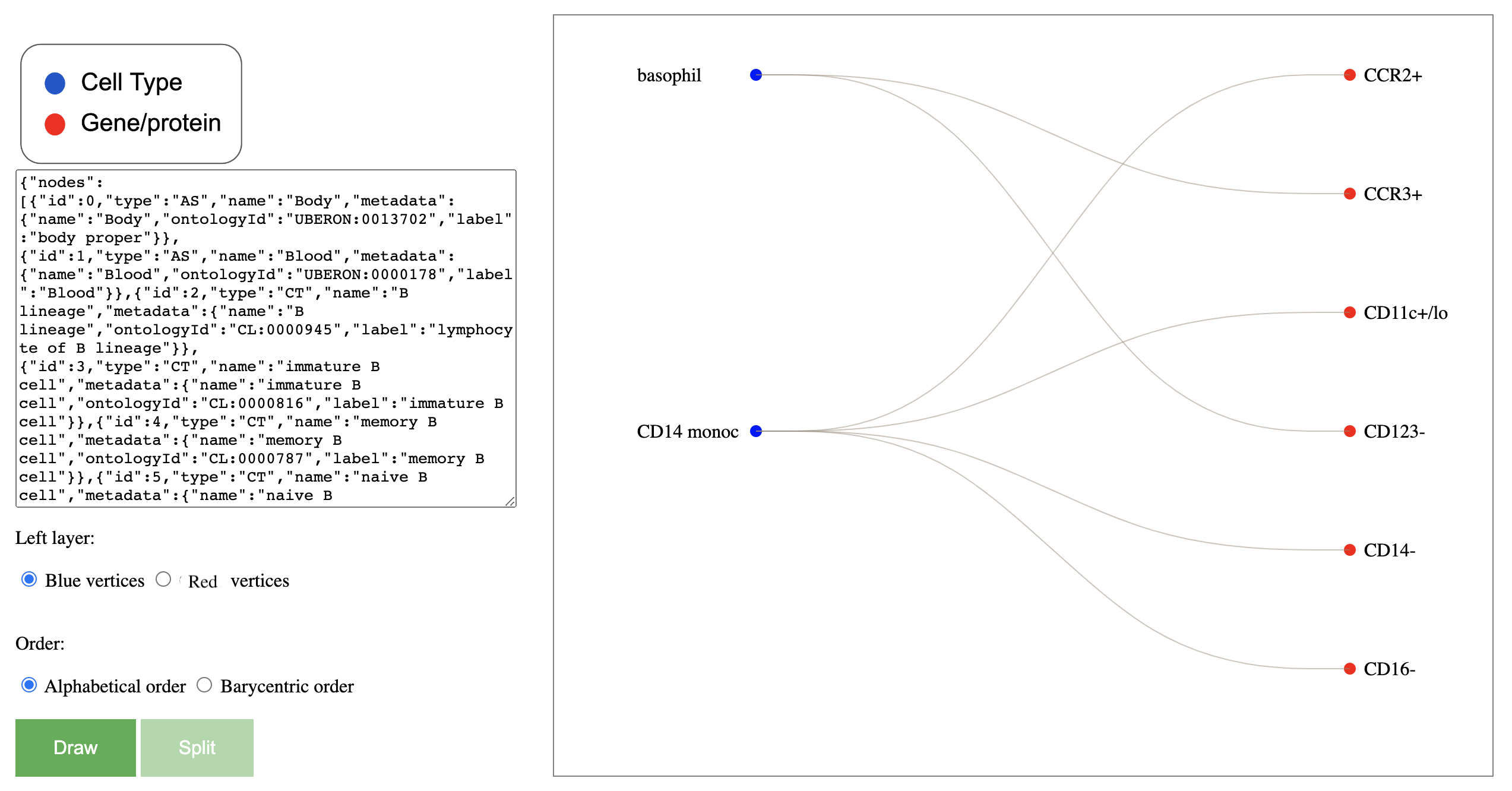}
    \caption{The input layout on the right side appears after inserting the dataset into the text area and clicking the draw button.}
    \end{subfigure}
    \begin{subfigure}[b]{0.70\textwidth}
    \includegraphics[width=\textwidth]{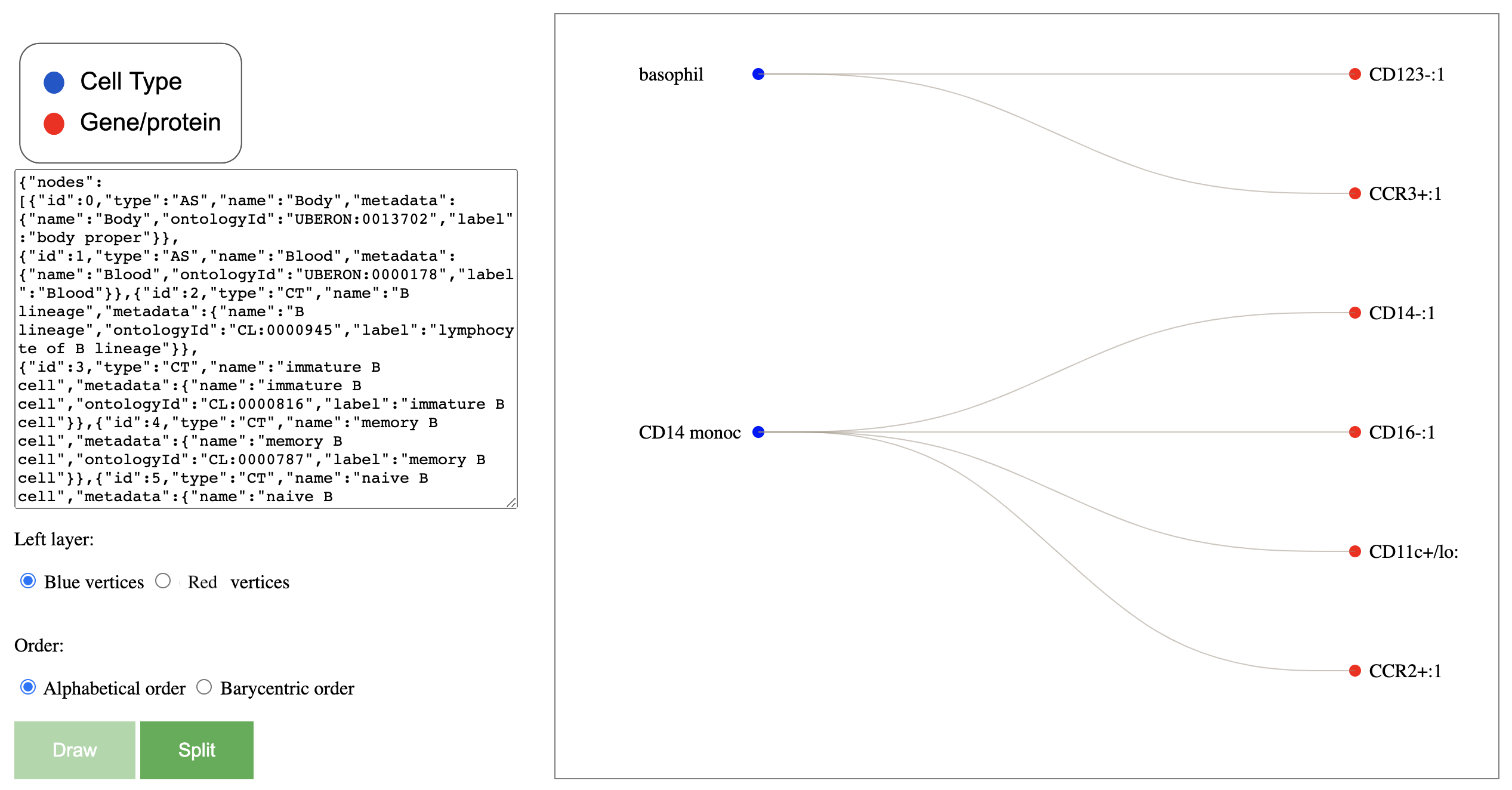}
    \caption{The output layout appears on the right side after clicking the split button.}
    \end{subfigure}
    \caption{The user interface. Datasets of the human body~\cite{github} can be uploaded, processed with the presented algorithms, and visualized.}
    \label{fig:interface}
\end{figure}
    
\subsection{Evaluation Results}
We first evaluate the exact algorithms for \crs{k} and \crsv{k} with Fixed Order. We have run experiments on 22 organ graphs on four settings: (1) the blue vertices (cell type) are fixed and the initial layout is generated using alphabetical order, (2) the red vertices (gene/protein) are fixed and the initial layout is generated using alphabetical order, (3) the blue vertices are fixed and initial layout is generated using barycentric heuristic, and (4) the red vertices are fixed and initial layout is generated using the barycentric heuristic. For each setting, we provide the initial number of crossings, number of vertices in the top (or left) layers that have fixed order, the number of bottom (or right) vertices, the number of splits, the number of split vertices, and the maximum number of splits; see Tables 2--5 in the supplemental material. 
The number of crossings in the initial layouts generated from alphabetical order is 2.7 times larger in total than in layouts generated by the barycentric heuristic. The number of splits is 1.58 times larger in total when we fix the gene/protein vertices (Tables 3 and 5). Note that for all organ graphs, the number of gene/protein vertices is relatively larger compared to the cell type vertices. When the cell-type vertices are fixed, there is more flexibility for splitting. Hence the numbers of splits in Tables 2 and 4 are smaller than in Tables 3 and 5. Similarly, the maximum number of splits is 2.5 times larger in total when the gene/protein vertices are fixed.

\begin{table}[H]
    \begin{center}
    \begin{tabular}{|m{5em}|c|m{.5em}|m{1.5em}|c|m{3.5em}|m{2.5em}|m{2.5em}|c|m{3.5em}|m{2.5em}|m{2.5em}|}
    \hline
     \multirow{2}{4em}{\textbf{Organ}} & \multirow{2}{4em}{\textbf{Crossings}} & \multirow{2}{4em}{\textbf{$|V_t|$}} & \multirow{2}{4em}{\textbf{$|V_b|$}} & \multicolumn{4}{c|}{\textbf{\crs{k} with Fixed Order}} & \multicolumn{4}{c|}{\textbf{\crsv{k} with Fixed Order}} \\ 
     \cline{5-12}
      &  &  &  & \textbf{Splits} & \textbf{Split vertices} & \textbf{Max splits} & \textbf{Time (ms)} & \textbf{Splits} & \textbf{Split vertices} & \textbf{Max splits} & \textbf{Time (ms)} \\ 
     \hline
     Bone Marrow & 111440 & 45 & 298 & 343 & 136 & 16 & 1.5 & 343 & 134 & 17 & 1.5 \\  
     \hline
     Brain & 28345 & 127 & 254 & 78 & 63 & 4 & 2.9 &  78 & 63 & 4 & 2.9 \\
     \hline
     Heart & 504 & 15 & 45 & 6 & 6 & 1 & 0.3 &  6 & 6 & 1 & 0.4 \\
     \hline
     Kidney & 13347 & 58 & 143 & 85 & 52 & 4 & 1.5 & 85 & 52 & 4 & 1.5 \\
     \hline
     Large intestine & 4778 & 51 & 73 & 58 & 25 & 6 & 0.5 &  59 & 24 & 6 & 0.6 \\
     \hline
     Lung & 11654 & 69 & 162 & 63 & 39 & 6 & 0.9 &  63 & 39 & 6 & 0.9 \\
     \hline
     Lymph nodes & 59709 & 44 & 255 & 213 & 100 & 10 & 1.0 & 213 & 100 & 10 & 1.0 \\
     \hline
     Skin & 2066 & 36 & 66 & 19 & 11 & 3 & 0.5 &  19 & 11 & 3 & 0.4 \\
     \hline
     Spleen & 40565 & 65 & 225 & 165 & 75 & 10 & 1.3 & 166 & 75 & 11 & 1.4 \\
     \hline
     Thymus & 102067 & 41 & 511 & 135 & 100 & 5 & 1.8 &  135 & 100 & 5 & 1.8 \\
     \hline
     Eye & 17046 & 47 & 98 & 164 & 78 & 5 & 0.8 &  166 & 78 & 5 & 0.9 \\
     \hline
     Fallopian Tube & 153 & 19 & 23 & 6 & 5 & 2 & 0.4 &  6 & 5 & 2 & 0.4 \\
     \hline
     Liver & 625 & 26 & 47 & 9 & 8 & 2 & 0.5 &  9 & 8 & 2 & 0.5 \\
     \hline
     Pancreas & 2510 & 29 & 40 & 56 & 32 & 6 & 0.5 &  56 & 32 & 6 & 0.6 \\
     \hline
     Peripheral Nervous System & 0 & 1 & 2 & 0 & 0 & 0 & 0.1 &  0 & 0 & 0 & 0.1 \\
     \hline
     Prostate & 405 & 12 & 31 & 3 & 3 & 1 & 0.3 &  3 & 3 & 1 & 0.3 \\
     \hline
     Ovary & 8 & 3 & 6 & 0 & 0 & 0 & 0.1 &  0 & 0 & 0 & 0.1 \\
     \hline
     Small Intestine & 28 & 5 & 13 & 0 & 0 & 0 & 0.1 &  0 & 0 & 0 & 0.1 \\
     \hline
     Ureter & 512 & 14 & 30 & 21 & 19 & 2 & 0.3 &  21 & 19 & 2 & 0.4 \\
     \hline
     Urinary Bladder & 628 & 15 & 31 & 21 & 21 & 1 & 1.1 &  21 & 21 & 1 & 1.1 \\
     \hline
     Uterus & 1147 & 16 & 45 & 20 & 12 & 4 & 0.2 &  20 & 12 & 4 & 0.2 \\
     \hline
     Blood & 49071 & 30 & 149 & 288 & 131 & 11 & 1.2 &  290 & 131 & 14 & 1.3 \\
     \hline
    \end{tabular}
    \end{center}
    \caption{The vertices $V_t$ are the cell types and the vertices $V_b$ are the genes/proteins. The vertices are alphabetically ordered.}
    \label{tab:blue_alphabetical}
\end{table}

\begin{table}[H]
    \begin{center}
    \begin{tabular}{|m{4em}|c|m{1.5em}|m{1.5em}|c|m{3.5em}|m{2.5em}|m{2.5em}|c|m{3.5em}|m{2.5em}|m{2.5em}|}
    \hline
     \multirow{2}{4em}{\textbf{Organ}} & \multirow{2}{4em}{\textbf{Crossings}} & \multirow{2}{4em}{\textbf{$|V_t|$}} & \multirow{2}{4em}{\textbf{$|V_b|$}} & \multicolumn{4}{c|}{\textbf{\crs{k} with Fixed Order}} & \multicolumn{4}{c|}{\textbf{\crsv{k} with Fixed Order}} \\ 
     \cline{5-12}
      &  &  &  & \textbf{Splits} & \textbf{Split vertices} & \textbf{Max splits} & \textbf{Time (ms)} & \textbf{Splits} & \textbf{Split vertices} & \textbf{Max splits} & \textbf{Time (ms)} \\  
     \hline
     Bone Marrow & 111440 & 298 & 45 & 569 & 45 & 22 & 1.7 & 569 & 45 & 22 & 1.6 \\  
     \hline
     Brain & 28345 & 254 & 127 & 214 & 124 & 4 & 3.1 &  214 & 124 & 4 & 3.1 \\
     \hline
     Heart & 504 & 45 & 15 & 30 & 14 & 4 & 0.3 &  30 & 14 & 4 & 0.3 \\
     \hline
     Kidney & 13347 & 143 & 58 & 164 & 56 & 7 & 1.4 &  164 & 56 & 7 & 1.3 \\
     \hline
     Large intestine & 4778 & 73 & 51 & 76 & 32 & 6 & 0.5 &  76 & 32 & 6 & 0.4 \\
     \hline
     Lung & 11654 & 162 & 69 & 151 & 64 & 4 & 1.1 &  151 & 64 & 4 & 1.1 \\
     \hline
     Lymph nodes & 59709 & 255 & 44 & 378 & 42 & 26 & 1.0 &  379 & 42 & 27 & 0.9 \\
     \hline
     Skin & 2066 & 66 & 36 & 48 & 32 & 2 & 0.6 &  48 & 32 & 2 & 0.6 \\
     \hline
     Spleen & 40565 & 225 & 65 & 311 & 58 & 20 & 1.5 &  312 & 58 & 19 & 1.5 \\
     \hline
     Thymus & 102067 & 511 & 41 & 514 & 36 & 69 & 2.1 &  514 & 36 & 69 & 2.1 \\
     \hline
     Eye & 17046 & 98 & 47 & 172 & 29 & 40 & 0.9 &  179 & 29 & 37 & 0.9 \\
     \hline
     Fallopian Tube & 153 & 23 & 19 & 12 & 8 & 2 & 0.5 &  12 & 8 & 2 & 0.5 \\
     \hline
     Liver & 625 & 47 & 26 & 27 & 12 & 4 & 0.7 &  27 & 12 & 4 & 0.6 \\
     \hline
     Pancreas & 2510 & 40 & 29 & 53 & 22 & 7 & 0.6 &  53 & 22 & 7 & 0.6 \\
     \hline
     Peripheral Nervous System & 0 & 2 & 1 & 0 & 0 & 0 & 0.1 &  0 & 0 & 0 & 0.1 \\
     \hline
     Prostate & 405 & 31 & 12 & 20 & 12 & 2 & 0.3 &  20 & 12 & 2 & 0.3 \\
     \hline
     Ovary & 8 & 6 & 3 & 3 & 3 & 1 & 0.1 &  3 & 3 & 1 & 0.1 \\
     \hline
     Small Intestine & 28 & 13 & 5 & 7 & 3 & 3 & 0.1 &  7 & 3 & 3 & 0.1 \\
     \hline
     Ureter & 512 & 30 & 14 & 28 & 13 & 6 & 0.5 &  28 & 13 & 6 & 0.5 \\
     \hline
     Urinary Bladder & 628 & 31 & 15 & 30 & 14 & 5 & 1.1 &  30 & 14 & 5 & 1.2 \\
     \hline
     Uterus & 1147 & 45 & 16 & 44 & 16 & 6 & 0.2 &  44 & 16 & 6 & 0.2 \\
     \hline
     Blood & 49071 & 149 & 30 & 379 & 30 & 40 & 1.4 &  379 & 30 & 40 & 1.4 \\
     \hline
    \end{tabular}
    \end{center}
    \caption{The vertices $V_t$ are the genes/proteins and the vertices $V_b$ are the cell types. The vertices are alphabetically ordered.}
    \label{tab:green_alphabetical}
\end{table}

\begin{table}[H]
    \begin{center}
    \begin{tabular}{|m{4em}|c|m{1.5em}|m{1.5em}|c|m{3.5em}|m{2.5em}|m{2.5em}|c|m{3.5em}|m{2.5em}|m{2.5em}|}
    \hline
     \multirow{2}{4em}{\textbf{Organ}} & \multirow{2}{4em}{\textbf{Crossings}} & \multirow{2}{4em}{\textbf{$|V_t|$}} & \multirow{2}{4em}{\textbf{$|V_b|$}} & \multicolumn{4}{c|}{\textbf{\crs{k} with Fixed Order}} & \multicolumn{4}{c|}{\textbf{\crsv{k} with Fixed Order}} \\ 
     \cline{5-12}
      &  &  &  & \textbf{Splits} & \textbf{Split vertices} & \textbf{Max splits} & \textbf{Time (ms)} & \textbf{Splits} & \textbf{Split vertices} & \textbf{Max splits} & \textbf{Time (ms)} \\  
     \hline
     Bone Marrow & 32599 & 45 & 298 & 323 & 130 & 17 & 1.5 & 323 & 124 & 17 & 1.6 \\ 
     \hline
     Brain & 4773 & 127 & 254 & 71 & 60 & 3 & 3.2 & 72 & 59 & 3 & 3.1 \\
     \hline
     Heart & 211 & 15 & 45 & 6 & 6 & 1 & 0.4 &  6 & 6 & 1 & 0.3 \\
     \hline
     Kidney & 1207 & 58 & 143 & 69 & 46 & 4 & 1.4 &  70 & 45 & 4 & 1.5 \\
     \hline
     Large intestine & 1878 & 51 & 73 & 49 & 25 & 6 & 0.7 &  50 & 24 & 7 & 0.6 \\
     \hline
     Lung & 1970 & 69 & 162 & 64 & 37 & 7 & 1.3 &  64 & 37 & 7 & 1.2 \\
     \hline
     Lymph nodes & 33315 & 44 & 255 & 214 & 104 & 10 & 1.3 &  216 & 104 & 10 & 1.2 \\
     \hline
     Skin & 339 & 36 & 66 & 14 & 9 & 3 & 0.8 &  14 & 9 & 3 & 0.7 \\
     \hline
     Spleen & 24833 & 65 & 225 & 171 & 76 & 9 & 1.7 &  172 & 75 & 9 & 1.8 \\
     \hline
     Thymus & 34863 & 41 & 511 & 136 & 102 & 5 & 2.3 &  137 & 101 & 5 & 2.4 \\
     \hline
     Eye & 8576 & 47 & 98 & 154 & 73 & 5 & 1.1 &  155 & 72 & 5 & 1.2 \\
     \hline
     Fallopian Tube & 47 & 19 & 23 & 7 & 7 & 1 & 0.7 &  7 & 7 & 1 & 0.6 \\
     \hline
     Liver & 84 & 26 & 47 & 8 & 7 & 2 & 0.9 &  8 & 7 & 2 & 0.9 \\
     \hline
     Pancreas & 925 & 29 & 40 & 46 & 29 & 5 & 0.8 &  46 & 29 & 5 & 0.7 \\
     \hline
     Peripheral Nervous System & 0 & 1 & 2 & 0 & 0 & 0 & 0.1 &  0 & 0 & 0 & 0.1 \\
     \hline
     Prostate & 64 & 12 & 31 & 3 & 3 & 1 & 0.4 &  3 & 3 & 1 & 0.4 \\
     \hline
     Ovary & 0 & 3 & 6 & 0 & 0 & 0 & 0.1 &  0 & 0 & 0 & 0.1  \\
     \hline
     Small Intestine & 0 & 5 & 13 & 0 & 0 & 0 & 0.1 &  0 & 0 & 0 & 0.1 \\
     \hline
     Ureter & 126 & 14 & 30 & 18 & 18 & 1 & 0.6 &  18 & 18 & 1 & 0.7 \\
     \hline
     Urinary Bladder & 151 & 15 & 31 & 18 & 18 & 1 & 1.3 &  18 & 18 & 1 & 1.3 \\
     \hline
     Uterus & 135 & 16 & 45 & 13 & 10 & 3 & 0.3 &  14 & 10 & 3 & 0.3 \\
     \hline
     Blood & 19308 & 30 & 149 & 284 & 130 & 13 & 1.7 &  284 & 130 & 13 & 1.6 \\
     \hline
    \end{tabular}
    \end{center}
    \caption{The vertices $V_t$ are the cell types and the vertices $V_b$ are the genes/proteins. The vertices are ordered by repeatedly applying the barycentric heuristic on both sides.}
    \label{tab:blue_barycentric}
\end{table}

\begin{table}[H]
    \begin{center}
    \begin{tabular}{|m{4em}|c|m{1.5em}|m{1.5em}|c|m{3.5em}|m{2.5em}|m{2.5em}|c|m{3.5em}|m{2.5em}|m{2.5em}|}
    \hline
     \multirow{2}{4em}{\textbf{Organ}} & \multirow{2}{4em}{\textbf{Crossings}} & \multirow{2}{4em}{\textbf{$|V_t|$}} & \multirow{2}{4em}{\textbf{$|V_b|$}} & \multicolumn{4}{c|}{\textbf{\crs{k} with Fixed Order}} & \multicolumn{4}{c|}{\textbf{\crsv{k} with Fixed Order}} \\ 
     \cline{5-12}
      &  &  &  & \textbf{Splits} & \textbf{Split vertices} & \textbf{Max splits} & \textbf{Time (ms)} & \textbf{Splits} & \textbf{Split vertices} & \textbf{Max splits} & \textbf{Time (ms)} \\ 
     \hline
     Bone Marrow & 32599 & 298 & 45 & 412 & 45 & 15 & 1.8 & 412 & 45 & 14 & 1.7 \\  
     \hline
     Brain & 4773 & 254 & 127 & 86 & 78 & 3 & 3.2 &  86 & 70 & 3 & 3.2 \\
     \hline
     Heart & 211 & 45 & 15 & 16 & 11 & 3 & 0.4 &  16 & 11 & 3 & 0.4 \\
     \hline
     Kidney & 1207 & 143 & 58 & 94 & 49 & 5 & 1.5 &  94 & 47 & 4 & 1.5 \\
     \hline
     Large intestine & 1878 & 73 & 51 & 47 & 24 & 5 & 0.8 &  50 & 24 & 5 & 0.7 \\
     \hline
     Lung & 1970 & 162 & 69 & 74 & 49 & 3 & 1.3 &  74 & 47 & 3 & 1.3 \\
     \hline
     Lymph nodes & 33315 & 255 & 44 & 298 & 41 & 24 & 1.4 &  299 & 41 & 23 & 1.3 \\
     \hline
     Skin & 339 & 66 & 36 & 22 & 19 & 2 & 0.8 &  22 & 19 & 2 & 0.8 \\
     \hline
     Spleen & 24833 & 225 & 65 & 223 & 59 & 14 & 1.8 &  227 & 59 & 15 & 1.8 \\
     \hline
     Thymus & 34863 & 511 & 41 & 320 & 35 & 46 & 2.4 &  320 & 34 & 48 & 2.5 \\
     \hline
     Eye & 8576 & 98 & 47 & 146 & 25 & 34 & 1.1 &  149 & 25 & 35 & 1.1 \\
     \hline
     Fallopian Tube & 47 & 23 & 19 & 5 & 5 & 1 & 0.6 &  5 & 5 & 1 & 0.6 \\
     \hline
     Liver & 84 & 47 & 26 & 3 & 7 & 2 & 1.1 &  3 & 7 & 2 & 1.0 \\
     \hline
     Pancreas & 925 & 40 & 29 & 42 & 21 & 6 & 0.8 &  42 & 20 & 6 & 0.8 \\
     \hline
     Peripheral Nervous System & 0 & 2 & 1 & 0 & 0 & 0 & 0.1 &  0 & 0 & 0 & 0.1 \\
     \hline
     Prostate & 64 & 31 & 12 & 13 & 12 & 2 & 0.5 &  13 & 12 & 2 & 0.4 \\
     \hline
     Ovary & 0 & 6 & 3 & 0 & 0 & 0 & 0.1 &  0 & 0 & 0 & 0.1  \\
     \hline
     Small Intestine & 0 & 13 & 5 & 3 & 3 & 1 & 0.1 &  3 & 3 & 1 & 0.1 \\
     \hline
     Ureter & 126 & 30 & 14 & 18 & 13 & 4 & 0.7 &  18 & 9 & 4 & 0.7 \\
     \hline
     Urinary Bladder & 151 & 31 & 15 & 19 & 14 & 4 & 1.4 &  19 & 14 & 4 & 1.3 \\
     \hline
     Uterus & 135 & 45 & 16 & 23 & 13 & 4 & 0.3 &  25 & 13 & 4 & 0.4 \\
     \hline
     Blood & 19308 & 149 & 30 & 300 & 30 & 31 & 1.8 &  300 & 30 & 31 & 1.7 \\
     \hline
    \end{tabular}
    \end{center}
    \caption{The vertices $V_t$ are the genes/proteins and the vertices $V_b$ are the cell types. The vertices are ordered by repeatedly applying the barycentric heuristic on both sides.}
    \label{tab:green_barycentric}
\end{table}

\begin{figure}[H]
    \centering
    \includegraphics[page=2, width=\textwidth]{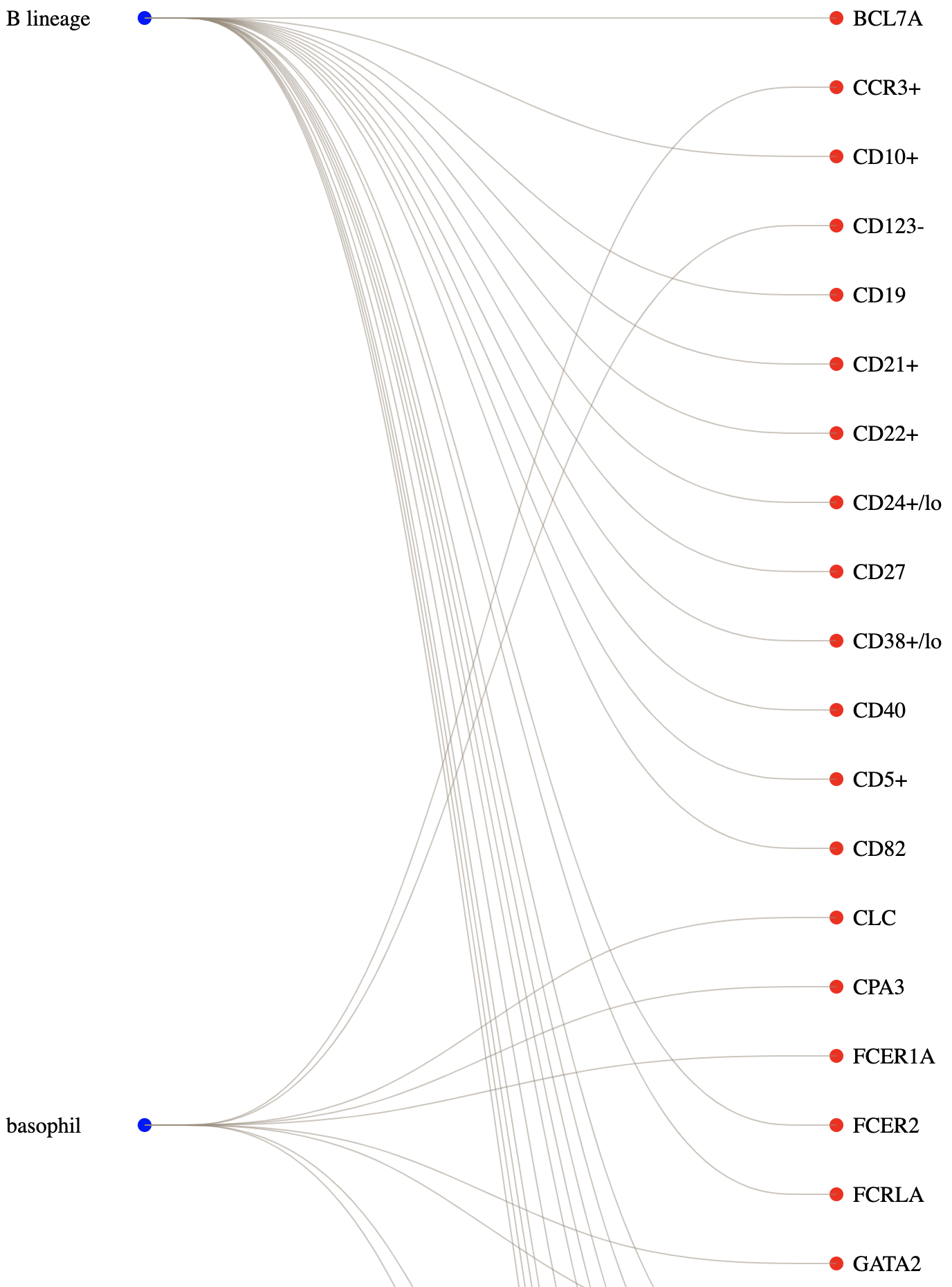}
    \caption{Part of the input layout of the 2-layer graph corresponding to the anatomical structure of blood computed by the exact algorithm for \crs{k} with Fixed Order. The graph is large, hence we show a part of the layout.}
\end{figure}

\begin{figure}[H]
    \centering
    \includegraphics[page=2, width=\textwidth]{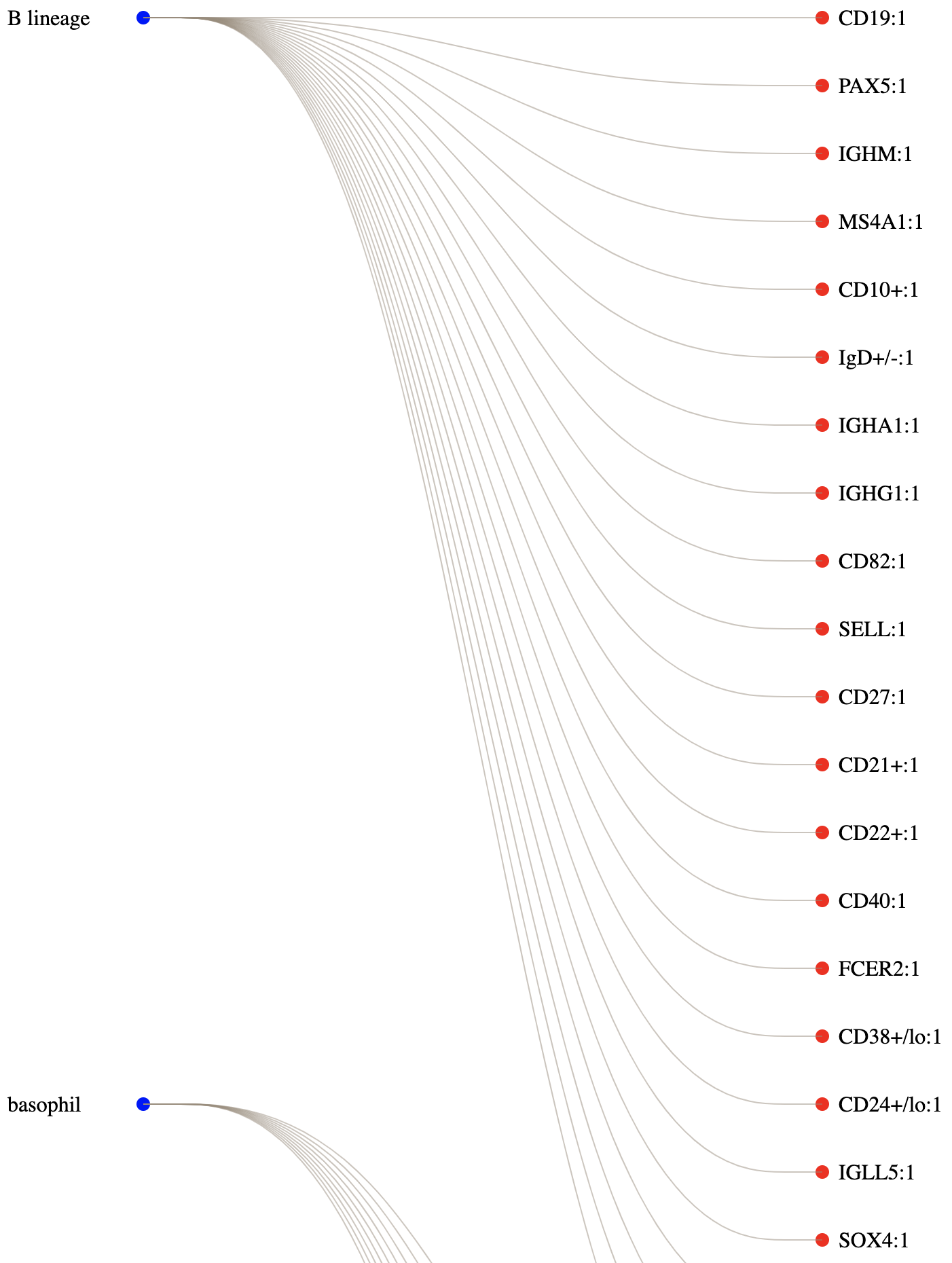}
    \caption{Part of the output layout of the 2-layer graph corresponding to the anatomical structure of blood computed by the exact algorithm for \crs{k} with Fixed Order.}
\end{figure}

\begin{figure}[H]
    \centering
    \begin{subfigure}[b]{0.475\textwidth}
        \centering
        \includegraphics[width=\textwidth]{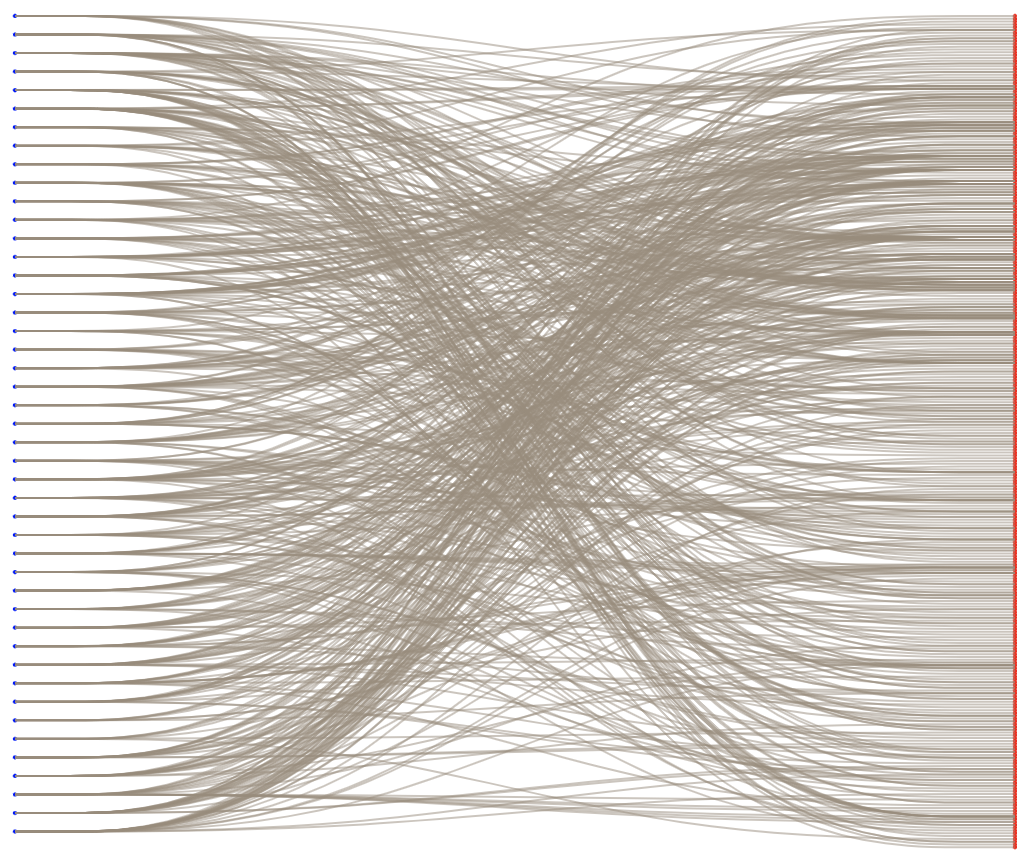}
        \caption[]%
        {{\small Input layout}}    
        \label{fig:mean and std of net14}
    \end{subfigure}
    \hfill
    \begin{subfigure}[b]{0.475\textwidth}  
        \centering 
        \includegraphics[width=\textwidth]{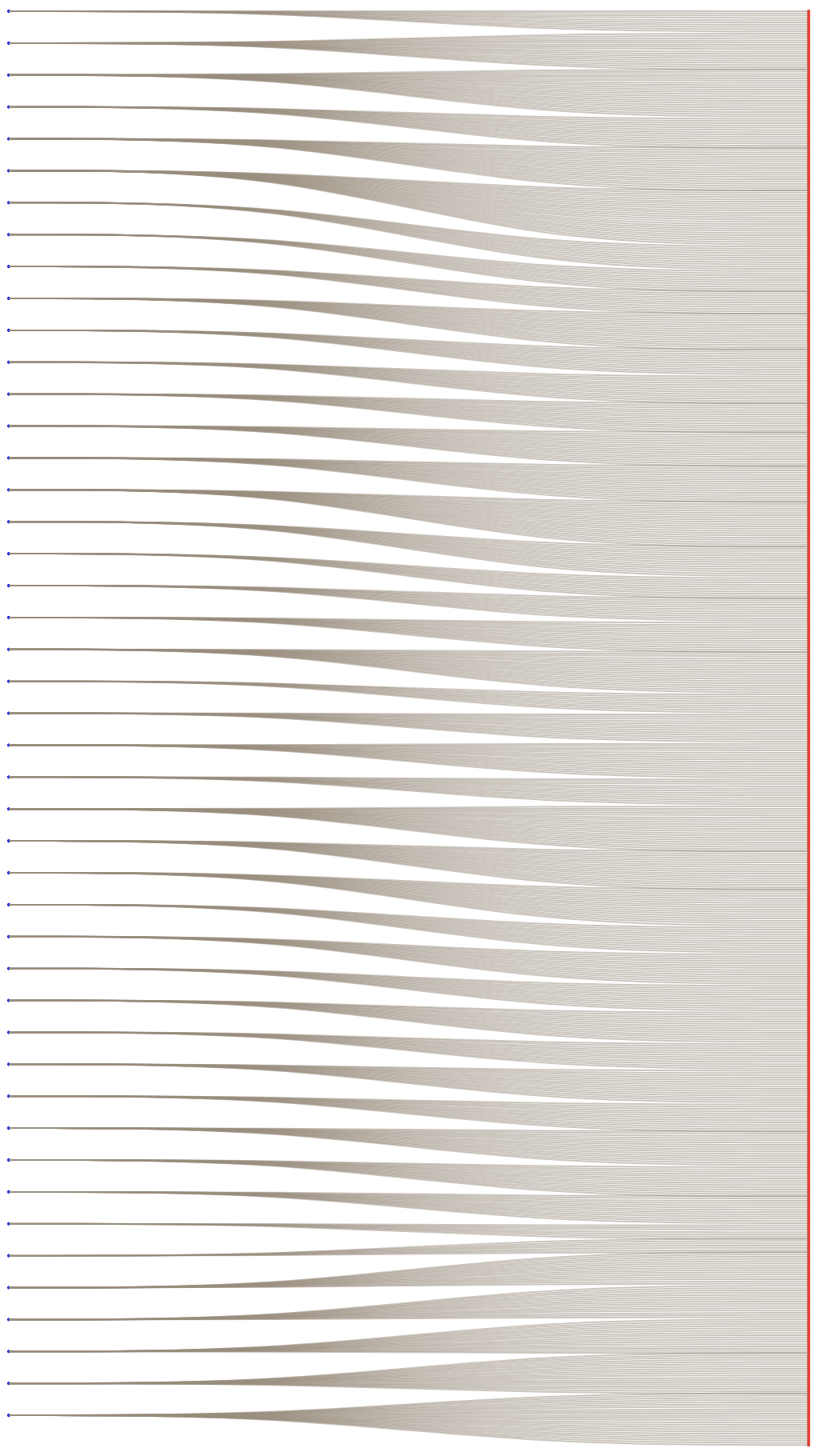}
        \caption[]%
        {{\small Output layout}}    
        \label{fig:mean and std of net24}
    \end{subfigure}
    \vskip\baselineskip
    \begin{subfigure}[b]{0.475\textwidth}   
        \centering 
        \includegraphics[width=\textwidth]{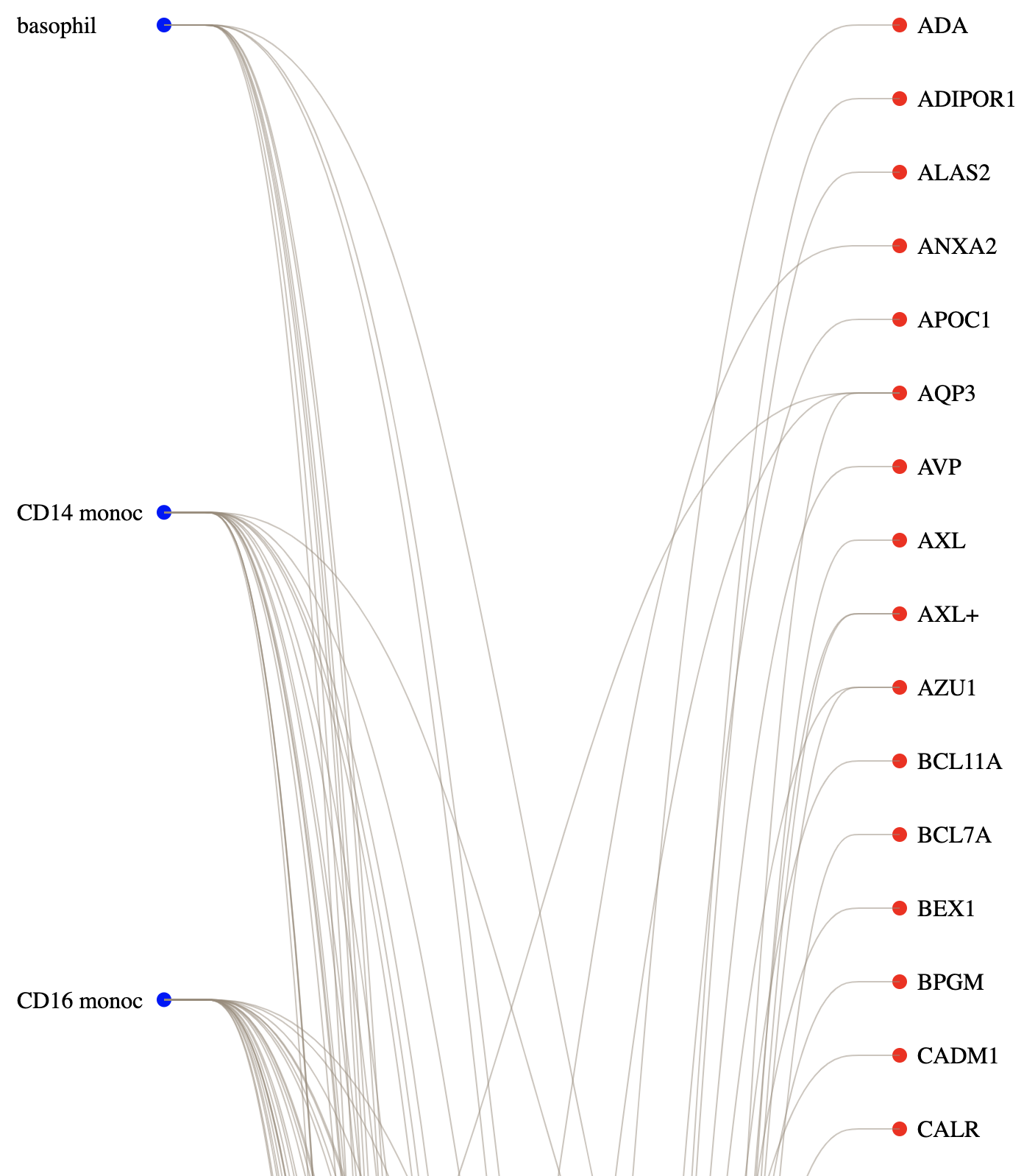}
        \caption[]%
        {{\small Zoomed in input layout}}    
        \label{fig:mean and std of net34}
    \end{subfigure}
    \hfill
    \begin{subfigure}[b]{0.475\textwidth}   
        \centering 
        \includegraphics[width=\textwidth]{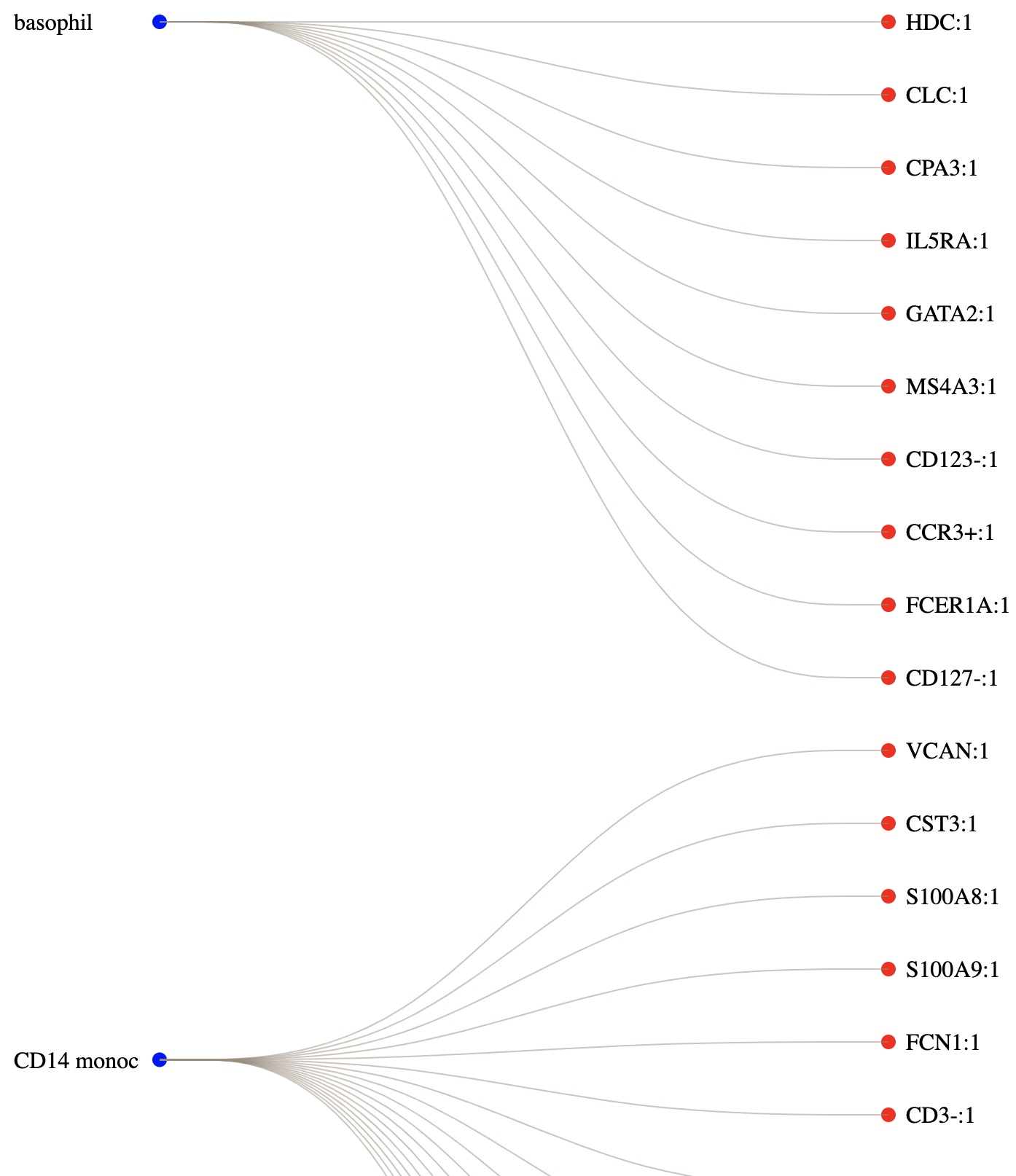}
        \caption[]%
        {{\small Zoomed in output layout}}    
        \label{fig:mean and std of net44}
    \end{subfigure}
    \caption[]
    {\small Input and output layouts of the exact algorithm of \crs{k} with Fixed Order for Bone Marrow.} 
    \label{fig:crs_bn_mr}
\end{figure}

\begin{figure}[H]
    \centering
    \begin{subfigure}[b]{0.475\textwidth}
        \centering
        \includegraphics[width=\textwidth]{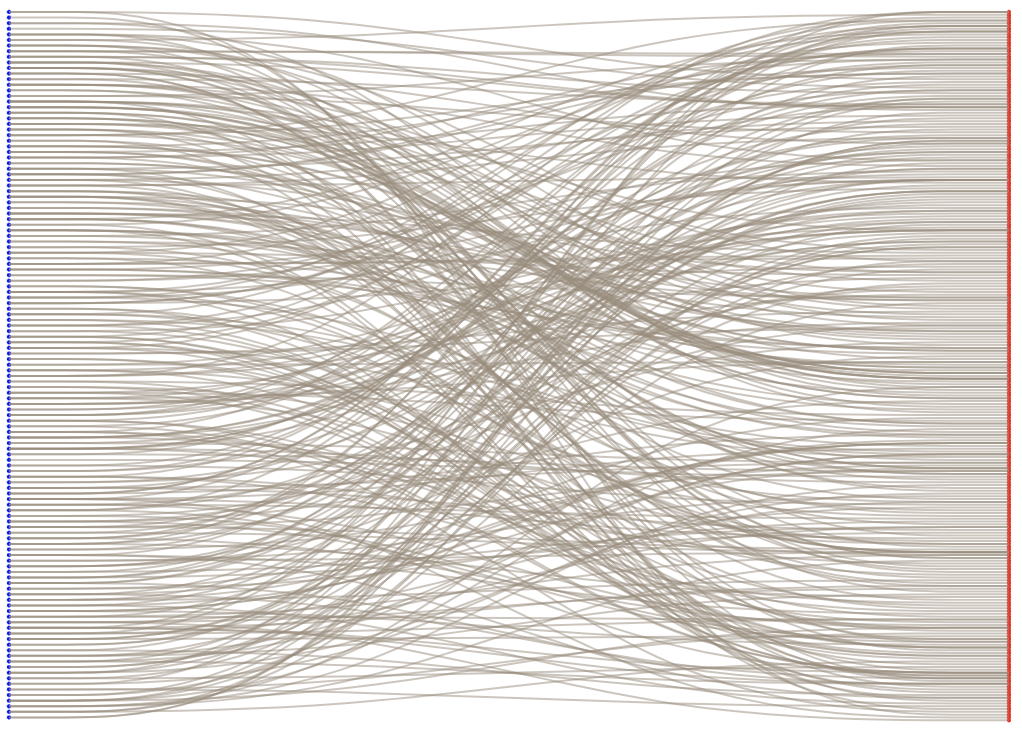}
        \caption[]%
        {{\small Input layout}}    
    \end{subfigure}
    \hfill
    \begin{subfigure}[b]{0.475\textwidth}  
        \centering 
        \includegraphics[width=\textwidth]{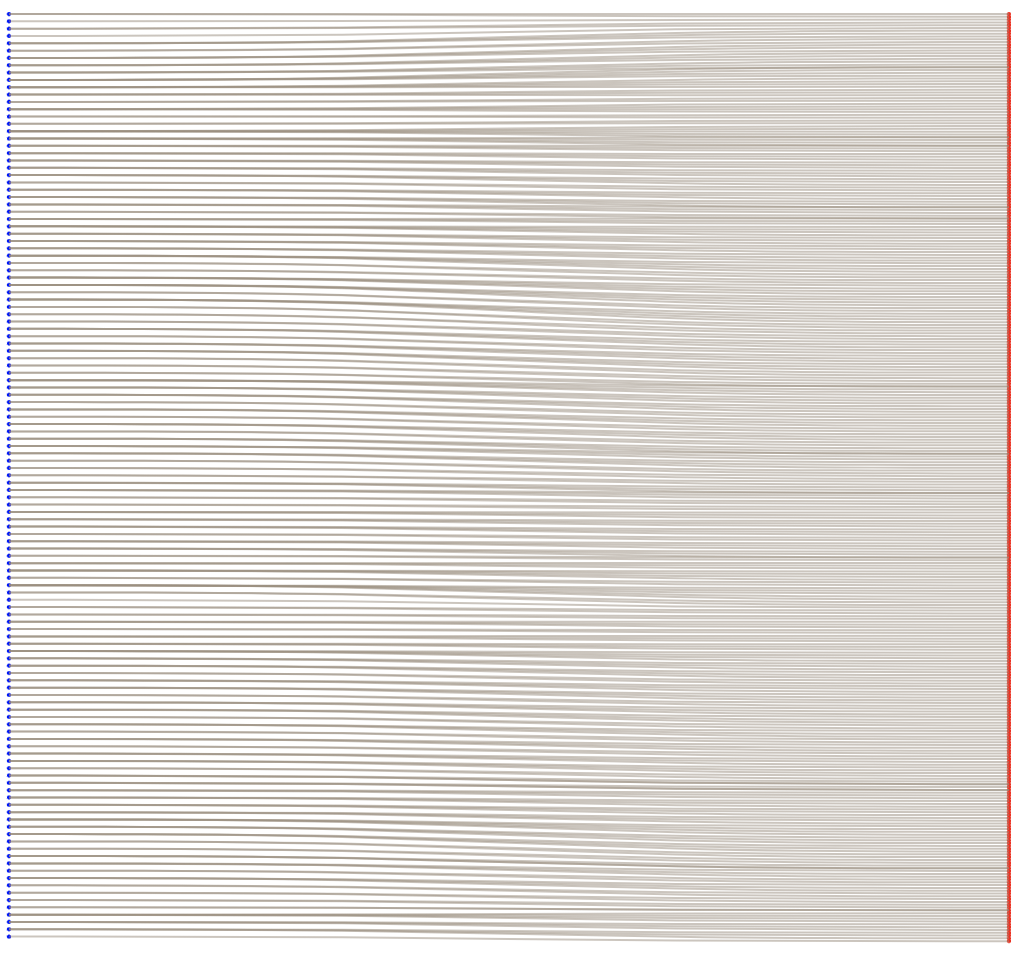}
        \caption[]%
        {{\small Output layout}}    
    \end{subfigure}
    \vskip\baselineskip
    \begin{subfigure}[b]{0.475\textwidth}   
        \centering 
        \includegraphics[width=\textwidth]{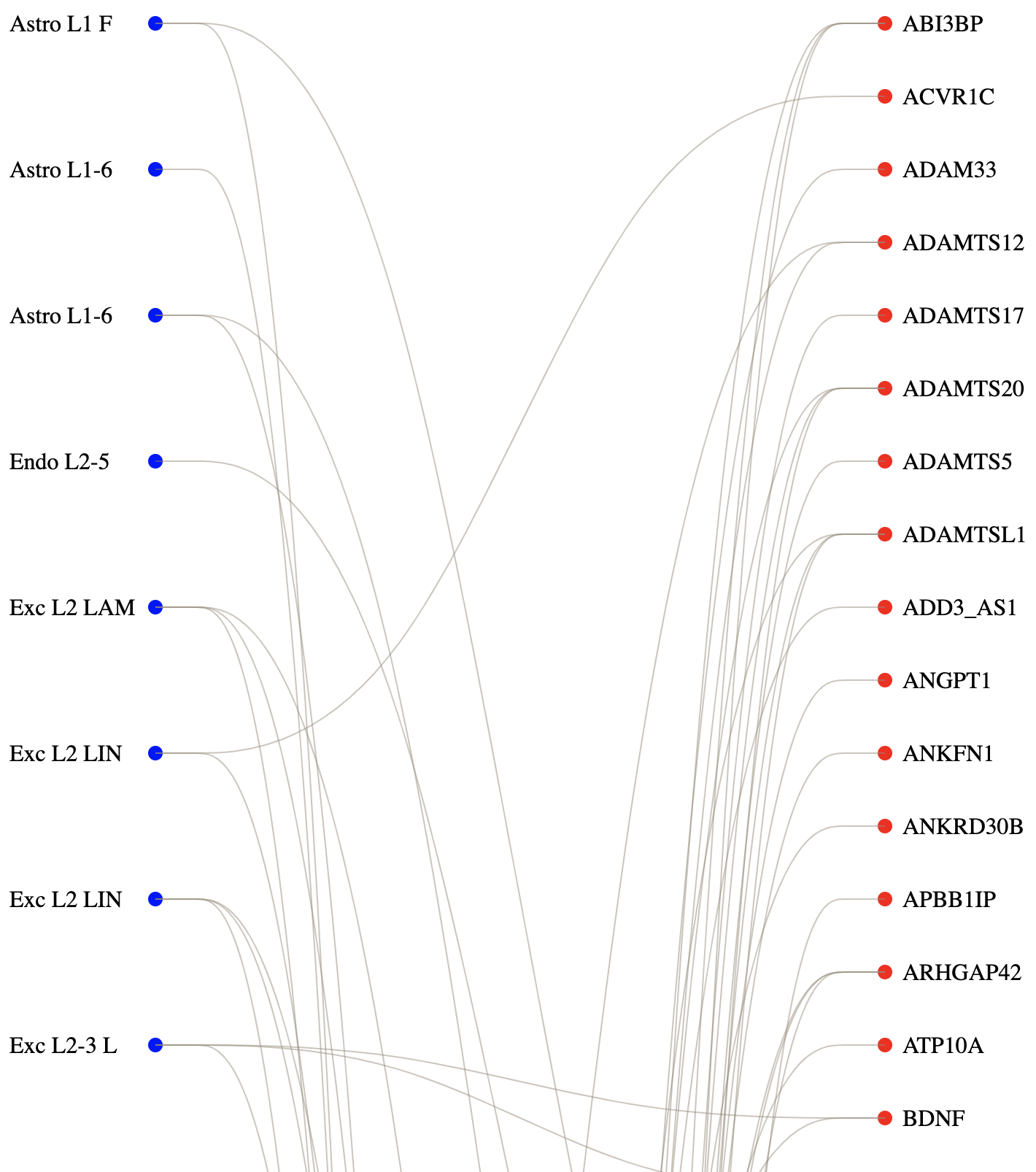}
        \caption[]%
        {{\small Zoomed in input layout}}    
    \end{subfigure}
    \hfill
    \begin{subfigure}[b]{0.475\textwidth}   
        \centering 
        \includegraphics[width=\textwidth]{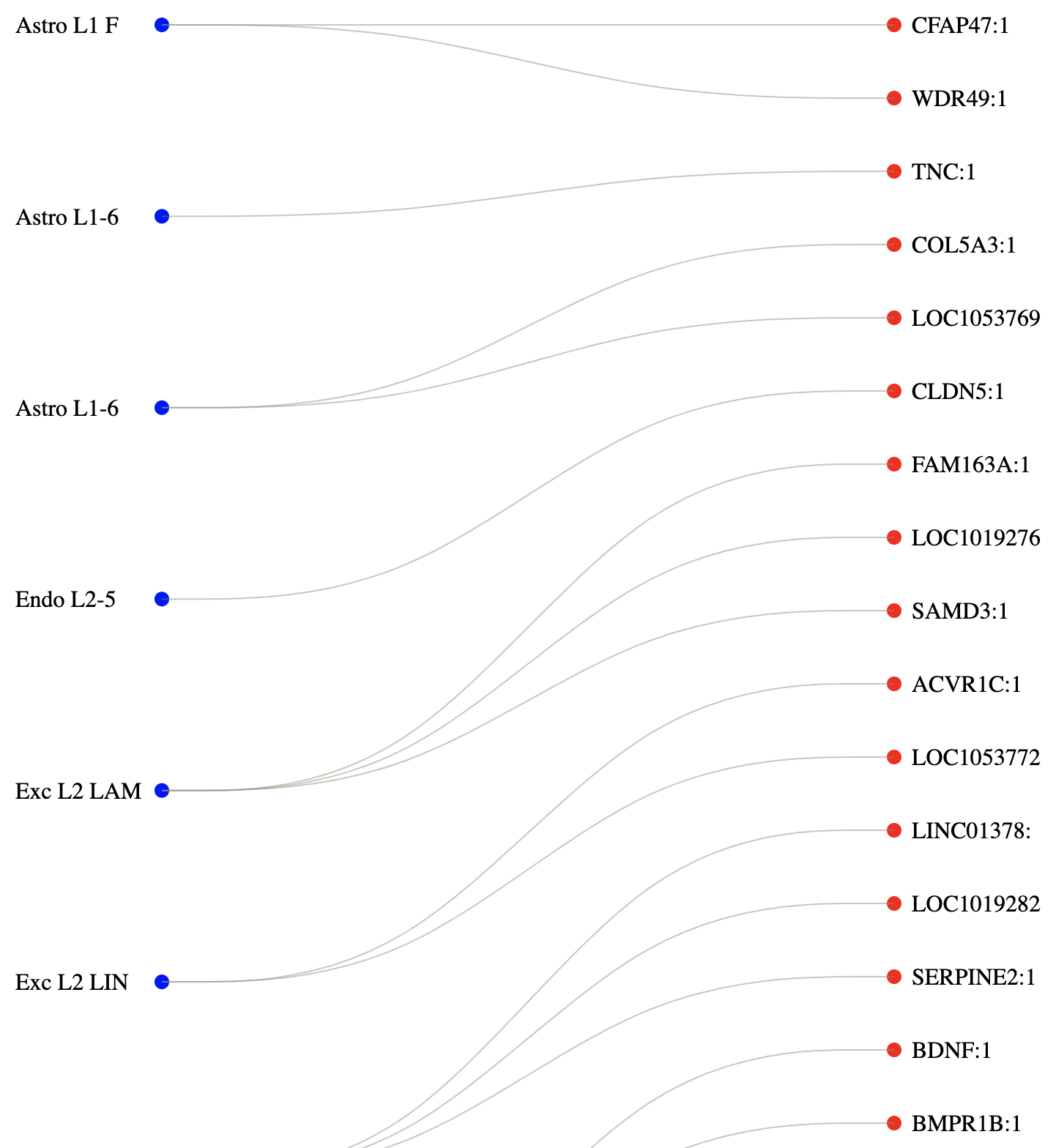}
        \caption[]%
        {{\small Zoomed in output layout}}    
    \end{subfigure}
    \caption[]
    {\small Input and output layouts of the exact algorithm of \crs{k} with Fixed Order for Brain.} 
    \label{fig:crs_brn}
\end{figure}

\begin{figure}[H]
    \centering
    \begin{subfigure}[b]{0.475\textwidth}
        \centering
        \includegraphics[width=\textwidth]{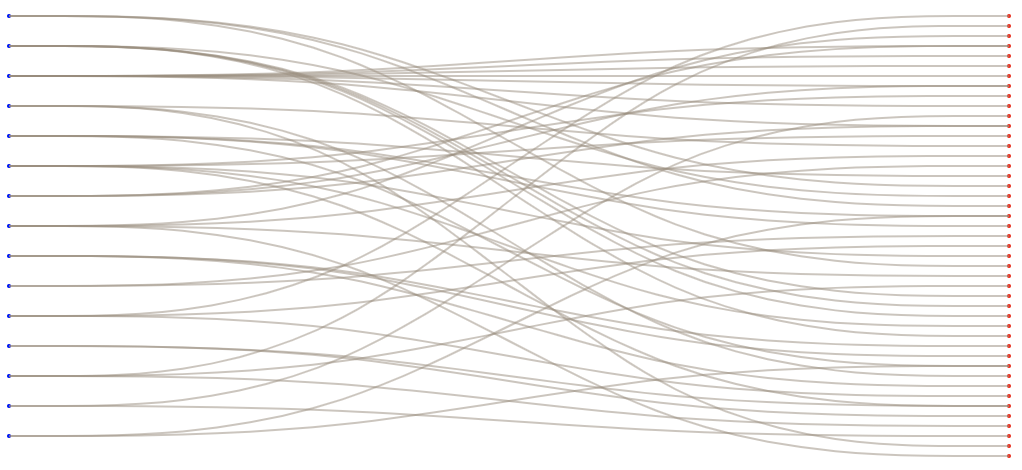}
        \caption[]%
        {{\small Input layout}}    
    \end{subfigure}
    \hfill
    \begin{subfigure}[b]{0.475\textwidth}  
        \centering 
        \includegraphics[width=\textwidth]{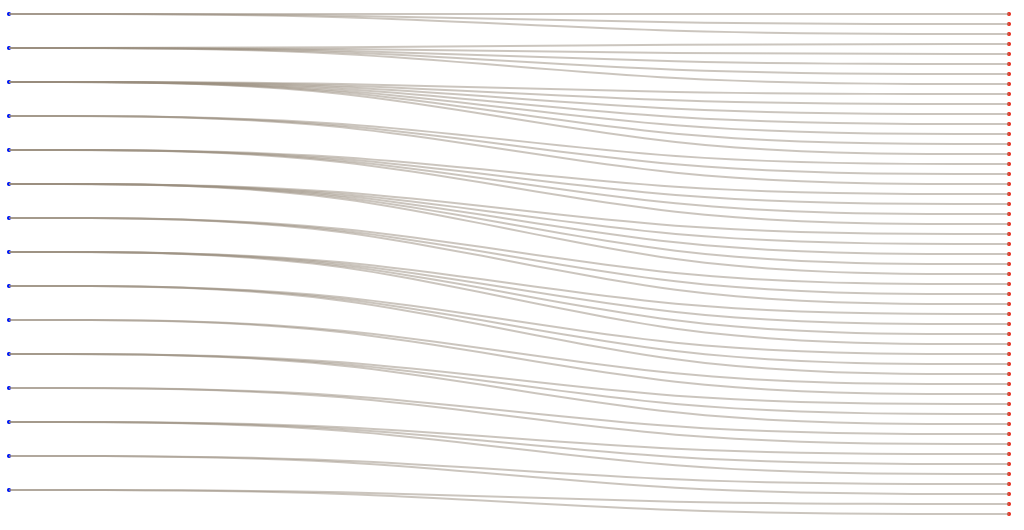}
        \caption[]%
        {{\small Output layout}}    
    \end{subfigure}
    \vskip\baselineskip
    \begin{subfigure}[b]{0.475\textwidth}   
        \centering 
        \includegraphics[width=\textwidth]{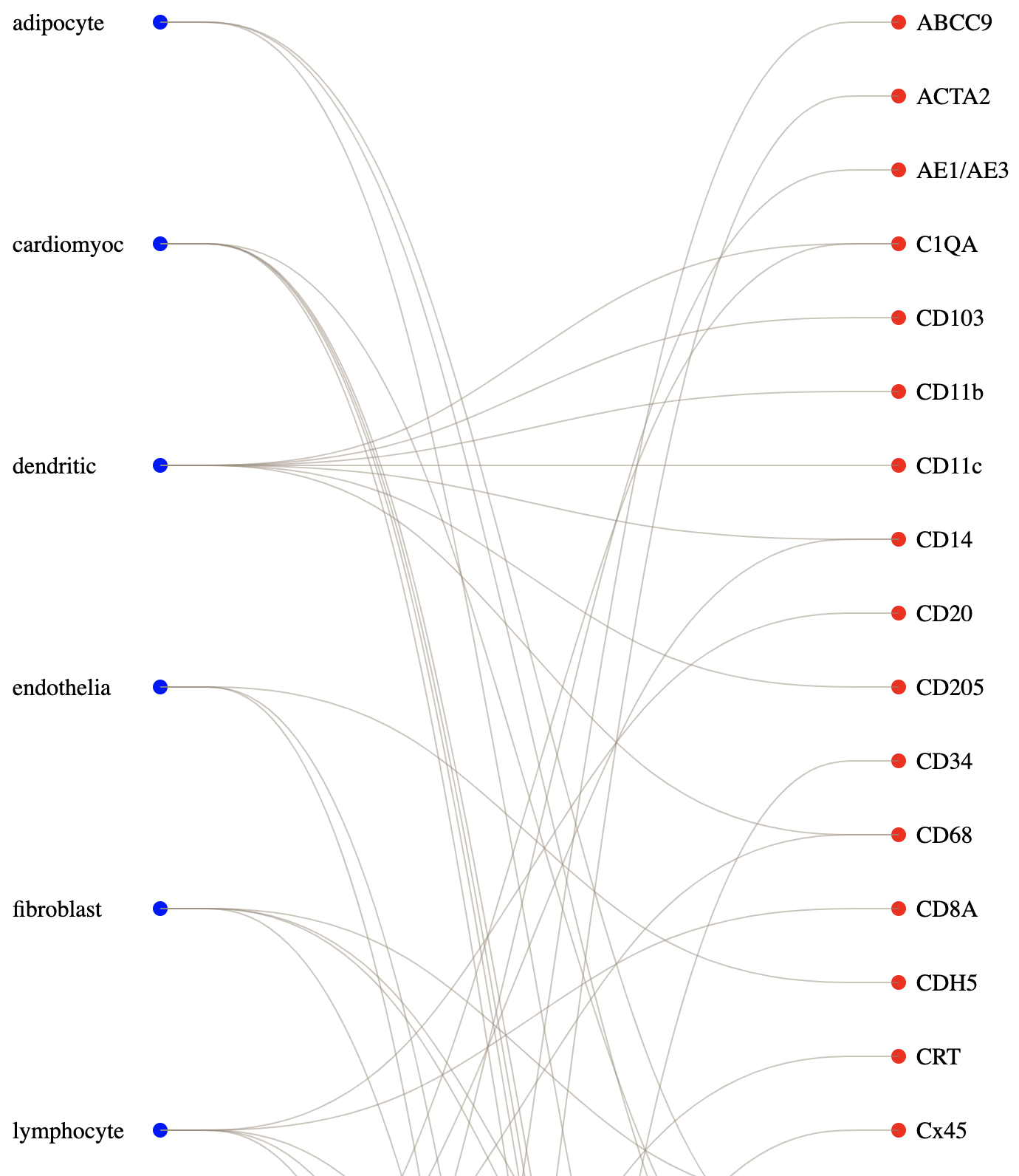}
        \caption[]%
        {{\small Zoomed in input layout}}    
    \end{subfigure}
    \hfill
    \begin{subfigure}[b]{0.475\textwidth}   
        \centering 
        \includegraphics[width=\textwidth]{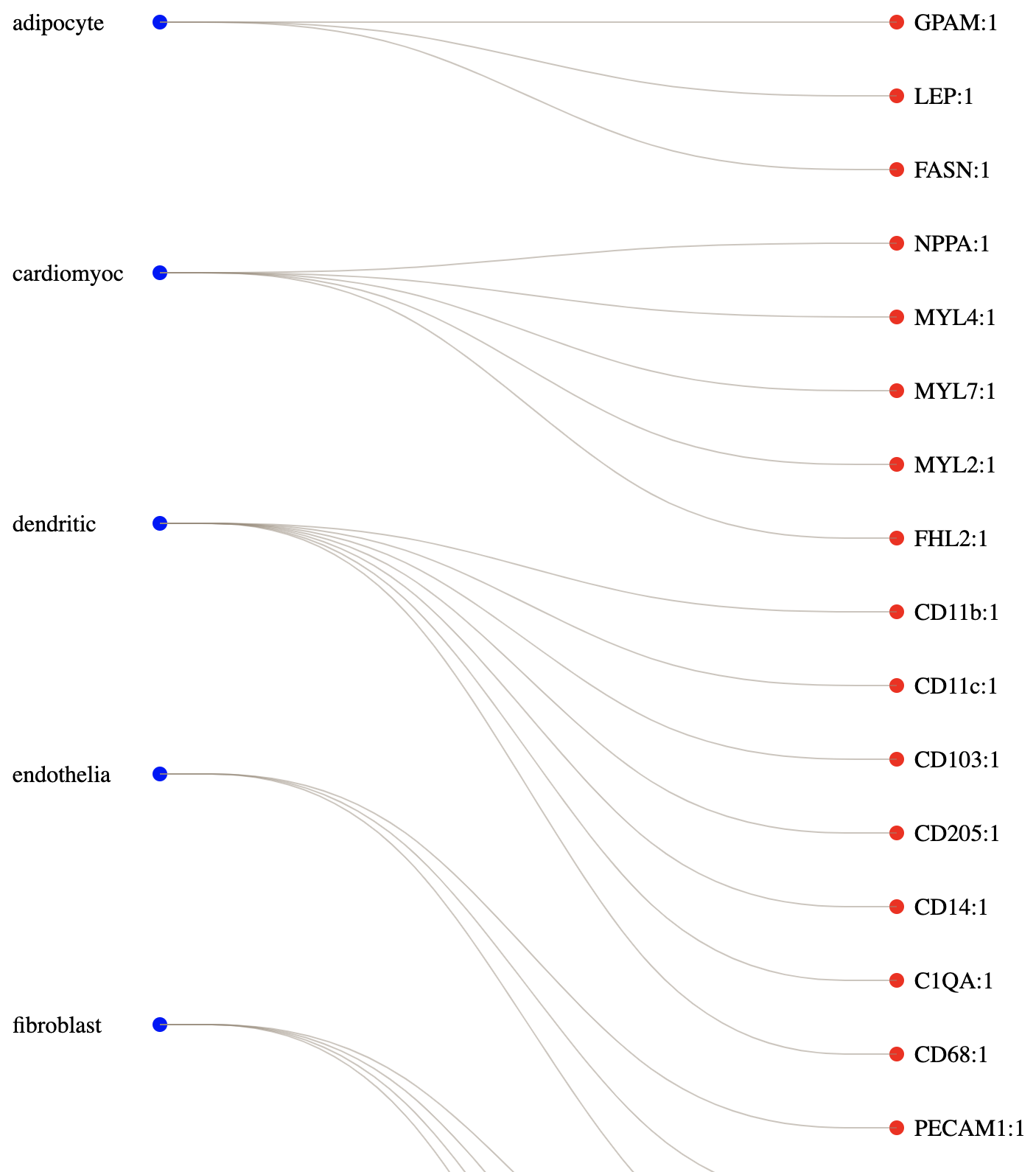}
        \caption[]%
        {{\small Zoomed in output layout}}    
    \end{subfigure}
    \caption[]
    {\small Input and output layouts of the exact algorithm of \crs{k} with Fixed Order for Heart.} 
    \label{fig:crs_hrt}
\end{figure}

A second set of experiments was conducted on the same 22 organ graphs to evaluate the crossing minimization heuristics. We set the maximum budget $k$ of splits to $200$ and computed the number of remaining crossings after each split. Additionally, we measured wall clock time after each iteration. Figure~\ref{fig:exp_1} shows the number of crossings in regards to $k$ for one example graphs. Examples of other graphs, as well as runtime plots can be found in the supplemental material; see Fig. 6--11. For both algorithms we observed a similar performance regarding crossing reduction. In some cases one algorithm slightly outperformed the other, but no clear trend is visible in the data. Intuitively, it seems that in the case of \textit{max-span} the length of edges correlates with the number of crossings. Furthermore, the number of crossings declines steeply at the beginning and for some graphs nearly 30\% of crossings are removed by the first 10 splits.

The runtime experiments confirmed the asymptotic runtime analysis. \textit{max-span} outperforms \textit{CR-count} on every datasets in regards to total runtime and 
scalability. 

\begin{figure}[H]
    \centering
    \begin{subfigure}[b]{0.47\textwidth}
    \includegraphics[width=\textwidth]{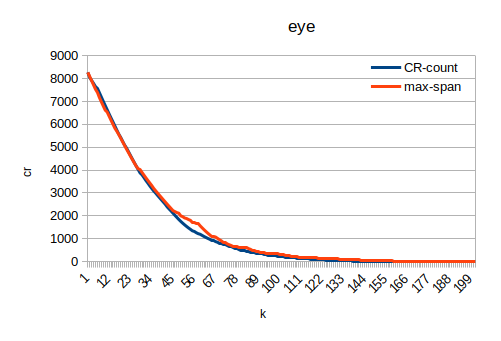}
    \caption{}
    \end{subfigure}
    \begin{subfigure}[b]{0.47\textwidth}
    \includegraphics[width=\textwidth]{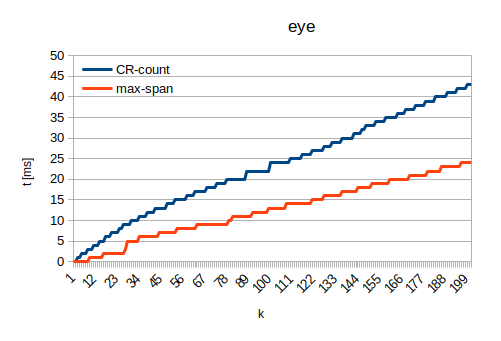}
    \caption{}
    \end{subfigure}
    \caption{\textit{max-span} and \textit{CR-count} heuristic applied to the eye dataset. Both algorithms have nearly identical performance regarding crossing reduction as seen in (a). However, \textit{max-span} outperforms \textit{CR-count} regarding runtime (b).}
    \label{fig:exp_1}
\end{figure}

\begin{figure}[H]
    \centering
    \begin{subfigure}[b]{0.47\textwidth}
    \includegraphics[width=\textwidth]{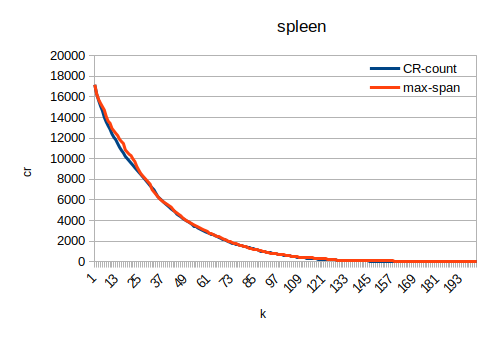}
    \caption{}
    \end{subfigure}
    \begin{subfigure}[b]{0.47\textwidth}
    \includegraphics[width=\textwidth]{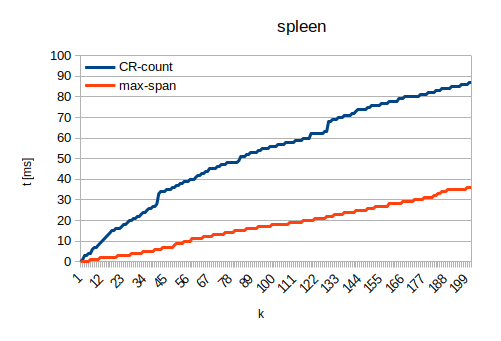}
    \caption{}
    \end{subfigure}
    \caption{\textit{max-span} and \textit{CR-count} heuristic applied to the spleen dataset. Both algorithms have nearly identical performance regarding crossing reduction as seen in (a). However, \textit{max-span} outperforms \textit{CR-count} regarding runtime (b).}
    \label{fig:exp_2}
\end{figure}

\begin{figure}[H]
    \centering
    \begin{subfigure}[b]{0.47\textwidth}
    \includegraphics[width=\textwidth]{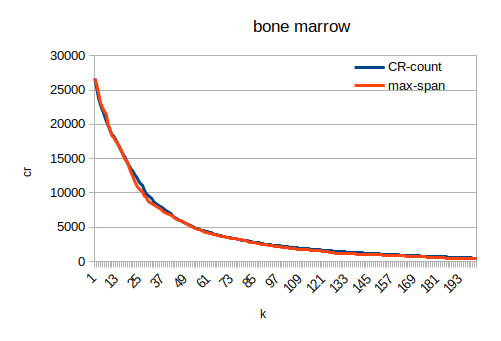}
    \caption{}
    \end{subfigure}
    \begin{subfigure}[b]{0.47\textwidth}
    \includegraphics[width=\textwidth]{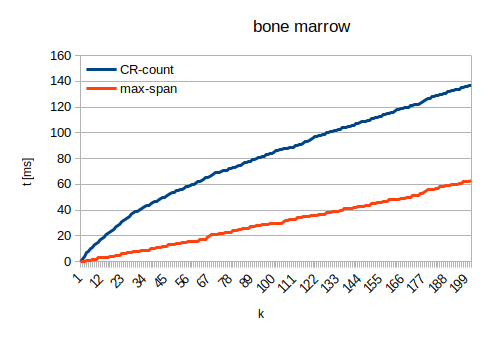}
    \caption{}
    \end{subfigure}
    \caption{\textit{max-span} and \textit{CR-count} heuristic applied to the bone marrow dataset. Both algorithms have nearly identical performance regarding crossing reduction as seen in (a). However, \textit{max-span} outperforms \textit{CR-count} regarding runtime (b).}
    \label{fig:exp_3}
\end{figure}

\begin{figure}[H]
    \centering
    \begin{subfigure}[b]{0.47\textwidth}
    \includegraphics[width=\textwidth]{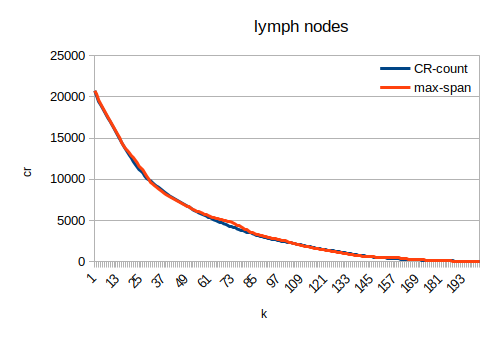}
    \caption{}
    \end{subfigure}
    \begin{subfigure}[b]{0.47\textwidth}
    \includegraphics[width=\textwidth]{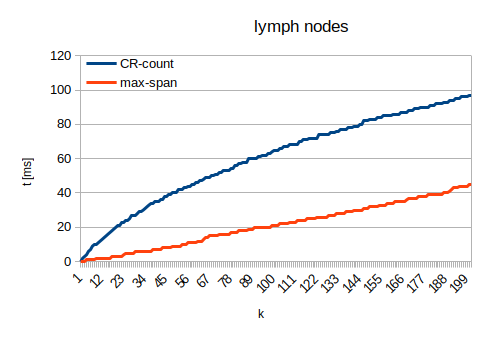}
    \caption{}
    \end{subfigure}
    \caption{\textit{max-span} and \textit{CR-count} heuristic applied to the lymph nodes dataset. Both algorithms have nearly identical performance regarding crossing reduction as seen in (a). However, \textit{max-span} outperforms \textit{CR-count} regarding runtime (b).}
    \label{fig:exp_4}
\end{figure}

\begin{figure}[H]
    \centering
    \begin{subfigure}[b]{0.47\textwidth}
    \includegraphics[width=\textwidth]{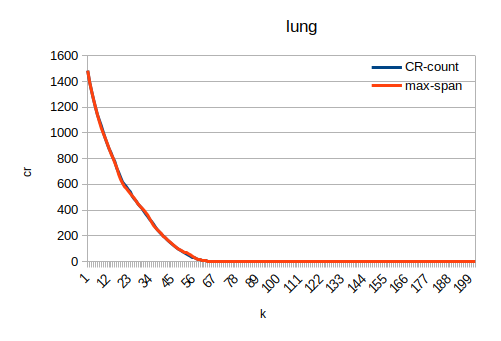}
    \caption{}
    \end{subfigure}
    \begin{subfigure}[b]{0.47\textwidth}
    \includegraphics[width=\textwidth]{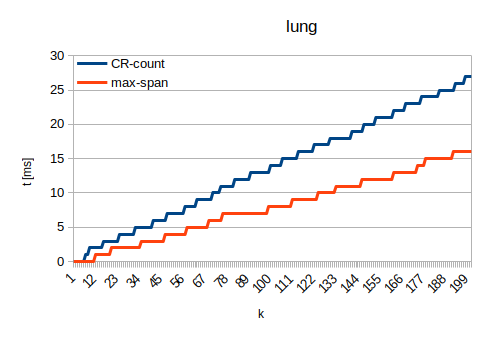}
    \caption{}
    \end{subfigure}
    \caption{\textit{max-span} and \textit{CR-count} heuristic applied to the lung nodes dataset. Both algorithms have nearly identical performance regarding crossing reduction as seen in (a). However, \textit{max-span} outperforms \textit{CR-count} regarding runtime (b).}
    \label{fig:exp_5}
\end{figure}

\begin{figure}[H]
    \centering
    \begin{subfigure}[b]{0.47\textwidth}
    \includegraphics[width=\textwidth]{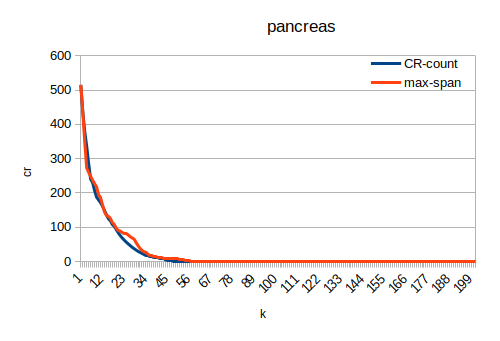}
    \caption{}
    \end{subfigure}
    \begin{subfigure}[b]{0.47\textwidth}
    \includegraphics[width=\textwidth]{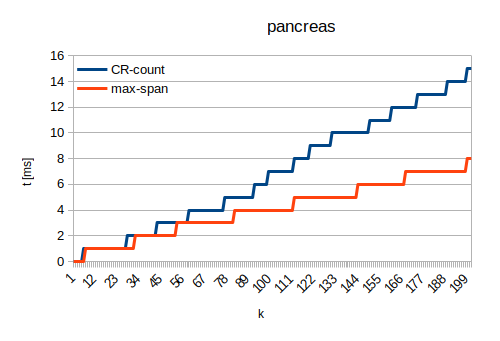}
    \caption{}
    \end{subfigure}
    \caption{\textit{max-span} and \textit{CR-count} heuristic applied to the pancrease nodes dataset. Both algorithms have nearly identical performance regarding crossing reduction as seen in (a). However, \textit{max-span} outperforms \textit{CR-count} regarding runtime (b).}
    \label{fig:exp_6}
\end{figure}

All codes for user interface, algorithms, experimental data, and analysis are available on Github at \url{https://github.com/abureyanahmed/split_graphs}.

\section{Open Problems}
\label{sec:conclusions}

Minimizing the total number of splits, or the number of split vertices are natural problems. Other variants include minimizing the maximum number of splits per vertex and considering the case  
where splits are allowed in both layers.
Vertex splits can also be used to improve other quality measures of a 2-layer layout (besides crossings). When visualizing large bipartite graphs, a natural goal is to arrange the vertices so that a small window can capture all the neighbors of a given node, i.e., minimize the maximum distance between the first and last neighbors of a top vertex in the order of the bottom vertices.


Since a great deal of vertex splitting can dramatically change the structure of the graph, it is desirable to have a bound on the number of splits. Hence, in our problems, we take a parameter $k$ that limits the number of splits. One can also consider the problem of finding the minimum number of splits to generate a planar graph. This is known as the Planar Split Thickness of Graphs~\cite{ekkllm-pstg-18}. It remains an interesting future research direction to find the trade-off between the number of splits and the number of crossings as the graph size changes.

\subsection{Acknowledgments.} This work started at Dagstuhl Seminar 21152 ``Multi-Level Graph Representation for Big Data Arising in Science Mapping''. We thank the organizers and participants for the discussions, particularly  C.~Raftopoulou.

\bibliographystyle{splncs04}
\bibliography{bib}

\begin{thebibliography}{10}
\providecommand{\url}[1]{\texttt{#1}}
\providecommand{\urlprefix}{URL }
\providecommand{\doi}[1]{https://doi.org/#1}

\bibitem{github}
\url{https://hubmapconsortium.github.io/ccf-asct-reporter/}

\bibitem{AhmedKK22}
Ahmed, R., Kobourov, S.G., Kryven, M.: An {FPT} algorithm for bipartite vertex
  splitting. In: Angelini, P., von Hanxleden, R. (eds.) Graph Drawing and
  Network Visualization - 30th International Symposium, {GD} 2022, Tokyo,
  Japan, September 13-16, 2022, Revised Selected Papers. Lecture Notes in
  Computer Science, vol. 13764, pp. 261--268. Springer (2022).
  \doi{10.1007/978-3-031-22203-0\_19},
  \url{https://doi.org/10.1007/978-3-031-22203-0\_19}

\bibitem{borner2021anatomical_bibtex}
B{\"o}rner, K., Teichmann, S.A., Quardokus, E.M., Gee, J.C., Browne, K.,
  Osumi-Sutherland, D., Herr, B.W., Bueckle, A., Paul, H., Haniffa, M., et~al.:
  Anatomical structures, cell types and biomarkers of the human reference
  atlas. Nature cell biology  \textbf{23}(11),  1117--1128 (2021)

\bibitem{DBLP:journals/tcad/ChaudharyCHNRW07}
Chaudhary, A., Chen, D.Z., Hu, X.S., Niemier, M.T., Ravichandran, R., Whitton,
  K.: Fabricatable interconnect and molecular {QCA} circuits. {IEEE} Trans.
  Comput. Aided Des. Integr. Circuits Syst.  \textbf{26}(11),  1978--1991
  (2007)

\bibitem{DemestrescuF01}
Demetrescu, C., Finocchi, I.: Removing cycles for minimizing crossings. JEA
  \textbf{6}, ~2 (2001)

\bibitem{dlm-sgdbp-19}
Didimo, W., Liotta, G., Montecchiani, F.: A survey on graph drawing beyond
  planarity. {ACM} Computing Surveys  \textbf{52}(1),  4:1--4:37 (2019)

\bibitem{emw-oacp-86}
Eades, P., McKay, B., Wormald, N.: On an edge crossing problem. In: ACSC '86.
  pp. 327--334 (1986)

\bibitem{em-vtl-95}
Eades, P., de~Mendonça~N., C.F.X.: Vertex-splitting and tension-free layouts.
  In: GD '95. vol.~1027, pp. 202--211. Springer (1996)

\bibitem{EadesW94}
Eades, P., Wormald, N.C.: Edge crossings in drawings of bipartite graphs.
  Algorithmica  \textbf{11}(4),  379--403 (1994)

\bibitem{ekkllm-pstg-18}
Eppstein, D., Kindermann, P., Kobourov, S., Liotta, G., Lubiw, A., Maignan, A.,
  Mondal, D., Vosoughpour, H., Whitesides, S., Wismath, S.: On the planar split
  thickness of graphs. Algorithmica  \textbf{80},  977--994 (2018)

\bibitem{faria_splitting_2001nourl}
Faria, L., de~Figueiredo, C.M.H., Mendonça, C.F.X.: Splitting number is
  {NP}-complete. DAM  \textbf{108}(1),  65--83 (2001)

\bibitem{HartsfieldJR85nourl}
Hartsfield, N., Jackson, B., Ringel, G.: The splitting number of the complete
  graph. Graphs and Combinatorics  \textbf{1}(1),  311--329 (1985)

\bibitem{Henry08_bibtex}
Henry, N., Bezerianos, A., Fekete, J.D.: Improving the readability of clustered
  social networks using node duplication. IEEE Transactions on Visualization
  and Computer Graphics  \textbf{14}(6),  1317--1324 (2008)

\bibitem{jackson_splitting_1984nourl}
Jackson, B., Ringel, G.: The splitting number of complete bipartite graphs.
  Archiv der Mathematik  \textbf{42}(2),  178--184 (1984)

\bibitem{jm-2scmpeha-97}
Jünger, M., Mutzel, P.: 2-layer straightline crossing minimization:
  Performance of exact and heuristic algorithms. JGAA  \textbf{1}(1),  1--25
  (1997)

\bibitem{Knauer2016}
Knauer, K., Ueckerdt, T.: Three ways to cover a graph. DM  \textbf{339}(2),
  745--758 (2016)

\bibitem{KobayashiT15nourl}
Kobayashi, Y., Tamaki, H.: A fast and simple subexponential fixed parameter
  algorithm for one-sided crossing minimization. Algorithmica  \textbf{72}(3),
  778--790 (2015)

\bibitem{Liebers-survey01}
Liebers, A.: Planarizing graphs -- a survey and annotated bibliography. JGAA
  \textbf{5}(1),  1--74 (2001)

\bibitem{mragkp-vamn-21}
McGee, F., Renoust, B., Archambault, D., Ghoniem, M., Kerren, A., Pinaud, B.,
  Pohl, M., Otjacques, B., Melançon, G., von Landesberger, T.: Visual Analysis
  of Multilayer Networks. Morgan \& Claypool (2021)

\bibitem{Nielsen19_bibtex}
Nielsen, S.S., Ostaszewski, M., McGee, F., Hoksza, D., Zorzan, S.: Machine
  learning to support the presentation of complex pathway graphs. IEEE/ACM
  transactions on computational biology and bioinformatics  \textbf{18}(3),
  1130--1141 (2019)

\bibitem{pfhlev-mvelbg-18}
Pezzotti, N., Fekete, J.D., Höllt, T., Lelieveldt, B.P.F., Eisemann, E.,
  Vilanova, A.: Multiscale visualization and exploration of large bipartite
  graphs. CGF  \textbf{37}(3),  549--560 (2018)

\bibitem{purchase97_bibtex}
Purchase, H.: Which aesthetic has the greatest effect on human understanding?
  In: Graph Drawing. vol.~97, pp. 248--261 (1997)

\bibitem{szh-trptn-11}
Scornavacca, C., Zickmann, F., Huson, D.H.: Tanglegrams for rooted phylogenetic
  trees and networks. Bioinformatics  \textbf{27}(13),  i248--i256 (2011)

\bibitem{SugiyamaTT81}
Sugiyama, K., Tagawa, S., Toda, M.: Methods for visual understanding of
  hierarchical system structures. TSMC  \textbf{11}(2),  109--125 (1981)

\end{thebibliography}


\end{document}